\documentclass{article}

\usepackage{arxiv}

\usepackage[utf8]{inputenc} 
\usepackage[T1]{fontenc}    
\usepackage{hyperref}       
\usepackage{url}            
\usepackage{booktabs}       
\usepackage{amsfonts}       
\usepackage{nicefrac}       
\usepackage{microtype}      
\usepackage{doi}
\usepackage{graphicx}
\usepackage{amsmath,amssymb}
\usepackage{amsthm}
\usepackage{caption}
\usepackage{color}
\usepackage{commath}
\usepackage{epsfig}
\usepackage{epstopdf}
\usepackage{lineno}
\usepackage{mathtools}
\usepackage{setspace}
\usepackage{subfigure}
\usepackage{wrapfig}
\usepackage{pdflscape}
\usepackage{rotating}
\usepackage{float}

\newtheorem{prop}{Proposition}

\title{Bifurcation analysis of the predator-prey model with the Allee effect in the predator}

\date{} 					

\author{Deeptajyoti Sen \\
	Department of Physical Sciences\\
	IISER Mohali\\
	Manauli, Punajb-140306 \\
	\texttt{deeps.sen.25@gmail.com} \\
	\And
	Andrew Morozov \thanks{Corresponding Author}\\
	Department of Mathematics\\
	University of Leicester\\
	Leicester, LE1 7RH \\
	\texttt{am379@leicester.ac.uk} \\
	\AND
	S. Ghorai \\
	Department of Mathematics \& Statistics \\
	IIT Kanpur \\
	Kalyanpur, Kanpur - 208016\\
	\texttt{sghorai@iitk.ac.in} \\
	\And
	Malay Banerjee \\
	Department of Mathematics \& Statistics\\
	IIT Kanpur \\
	Kalyanpur, Kanpur - 208016\\
    \texttt{malayb@iitk.ac.in} \\
}

\begin{document}
\maketitle
\begin{abstract}
The use of predator-prey models in theoretical ecology has a long history, and the model equations have largely evolved since the original Lotka-Volterra system towards more realistic descriptions of the processes of predation, reproduction and mortality. One important aspect is the recognition of the fact that the growth of a population can be subject to an Allee effect, where the per capita growth rate increases with the population density. Including an Allee effect has been shown to fundamentally change predator-prey dynamics and strongly impact species persistence, but previous studies mostly focused on scenarios of an Allee effect in the prey population. Here we explore a predator-prey model with an  ecologically important case of the Allee effect in the predator population where it occurs in the numerical response of predator without affecting its functional response. Biologically, this can result from various scenarios such as a lack of mating partners, sperm limitation and cooperative breeding mechanisms, among others. Unlike previous studies, we consider here a generic mathematical formulation of the Allee effect without specifying a concrete parameterisation of the functional form, and analyse the possible local bifurcations in the system. Further, we explore the global bifurcation structure of the model and its possible dynamical regimes for three different concrete parameterisations of the Allee effect. The model possesses a complex bifurcation structure: there can be multiple coexistence states including two stable limit cycles. Inclusion of the Allee effect in the predator generally has a destabilising effect on the coexistence equilibrium. We also show that regardless of the parametrisation of the Allee effect, enrichment of the environment will eventually result in extinction of the predator population.
\end{abstract}


\section{Introduction}
\label{intro}

Modelling predator-prey interactions has always been a mainstream area in mathematical biology and theoretical ecology. Our models have evolved tremendously since the famous Lotka-Volterra system, with one realistic modification being the introduction of non-monotonous per capita growth rates to the interacting species, as opposed to the monotonically decreasing per capita growth rate seen in the logistic equation. For instance, it is currently well recognised that the growth of natural populations can be subjected to the so-called Allee effect, where the per capita growth rate increases at low species densities  \cite{courchamp2008allee,fowler2002population}. The Allee effect can emerge at the population level due to a variety of mechanisms including enhancement in foraging efficiency, reproductive facilitation, collective defense and the modification of environmental conditions by organisms \cite{berec2007multiple,courchamp2008allee,fowler2002population}. There exist two types of Allee effect: weak and strong Allee effects. The weak Allee effect describes situations in which the per capita growth rate is increasing at small densities, but which nonetheless remains positive for low, nonzero population densities, while a strong Allee effect is characterized by a negative population growth at low densities since reproduction cannot compensate mortality rate. It has been demonstrated that including the Allee effect in predator-prey models has a strong impact on dynamics, in particular promoting population collapse and a further species extinction \cite{boukal2007predator,hilker2010population,lewis1993allee,morozov2006spatiotemporal,sen2012bifurcation}. In previous theoretical works, however, the main focus has been the scenario where there is an Allee effect in the growth rate of the prey rather than that of the predator. The scenarios where the predator growth is subject to the Allee effect are explored in the literature only partially. The aim of this paper is to contribute to bridging the gap.

The existing literature on the Allee effect in predators is scarce, and mainly focused on foraging facilitation among predators which occurs as a result of cooperative hunting  \cite{alves2017hunting,berec2010impacts,cosner1999effects,senallee}. Mathematically, this implies that the functional response of the predator is an increasing function of the predator density. In particular, it was shown that it might be detrimental for cooperative hunters to be too efficient in catching prey since this may cause resource over-exploitation and eventual extinction of the predator \cite{alves2017hunting,senallee}. On the other hand, the Allee effect can occur in predators due to other mechanisms such as low fertilization efficiency, a lack of mating partners, sperm limitation and cooperative breeding mechanisms \cite{berec2007multiple,courchamp2008allee,dennis1989allee}. From the modelling point of view, including an Allee effect in this case should affect the numerical response of the predator, since the food conversion efficiency becomes an increasing function of predator density,  while the functional response remains unchanged. As such, the model properties and ecological predictions will be different compared to the case of the foraging facilitation scenario. Some studies have considered the Allee effect in predators due to non-foraging mechanisms,  but none of them have been studied exhaustively in terms of the bifurcation structure, possible dynamical regimes and the role of parameterisations of the Allee effect in the model equations \cite{costa2018multiple,zhou2005stability}. The latter problem may be a general issue in ecological modelling and is related to so-called structural sensitivity, which is briefly described below.

In many ecological models, predator-prey systems in particular, there is often an uncertainty regarding which precise mathematical formulation of the model functions we need to implement in the model equations \cite{adamson2013can}. It is often impossible to determine which particular function we need to use in the model equations to describe predation, growth, mortality, competition, etc. Several parameterisations can fit available empirical data well, and different mathematical formulations can have a valid biological rationale \cite{flora2011structural}. Furthermore, implementation of close mathematical functions (both in terms of functional forms and their derivatives) in the same predator-prey model may result in different outcomes, in particular in topologically distinct bifurcation structures yielding different dynamical regimes \cite{adamson2013can,adamson2014bifurcation}. This property is called the structural sensitivity of biological models \cite{adamson2013can,adamson2014bifurcation,adamson2014defining,seo2018sensitivity}. Structural sensitivity may cause major problems in terms of generality of results obtained using specific concrete formulations of model functions such as growth rates or functional responses \cite{adamson2014bifurcation,aldebert2019three}. A possible way to address structural sensitivity is to allow for an unspecified formulation of some functions in the model equations with other functions being fixed, an approach is known as partially specified modelling \cite{wood1999super}. Implementation of the partially specified models approach is especially relevant for systems with the Allee effect in predators since this phenomenon is often caused by a variety of mechanisms, and is thus hard to describe by a single universal functional relation \cite{courchamp2008allee}. Moreover, the Allee effect can depend on the spatial scale of modelling, in which case the use of a single specific mathematical formulation for the dependence of the numerical response on the overall predator density is highly questionable \cite{courchamp2008allee}.

In this paper we explore a predator-prey model with an Allee effect in the predator which affects the numerical response of the predator without affecting its functional response. We consider a partially specified model, where the mathematical formulation of a strong Allee effect has only a few generic constraints to its shape. We explore the bifurcation structure of the model including saddle-node, Hopf, generalised Hopf and Bogdanov-Takens bifurcations of co-dimensions two and three. Then we construct and compare full bifurcation portraits obtained for three possible parameterisations of the Allee effect: the hyperbolic (Monod), exponential (Ivlev) and trigonometric formulations. We demonstrate that the model may exhibit structural sensitivity with respect to parameterisation of the Allee effect function. We find that adding the Allee effect  results in emergence of multiple non-trival attractors in the system which can potentially explain some empirically observed alternative states in ecosystems. We argue that the Allee effect in the predator growth has a large destabilising effect on population dynamics, which has been somehow neglected previously.

\section{Model Formulation} \label{sec:2}

We consider a Gause type prey-predator ODE model with a specialist predator \cite{Bazykin1998,Hsu2001,Kuang1988,Turchin2003}. The model equations read as follows
\begin{subequations} \label{eq:1}
\begin{align}
\begin{split}
\frac{dN}{dT}&=Nf(N)-g(N)P,
\end{split}\\
\begin{split}
\frac{dP}{dT}&=e\psi(P)g(N)P - \mu P,
\end{split}
\end{align}
\end{subequations}
where $N$ and $P$ are the population densities of prey and predator, respectively, at time $T$.

The function $f(N)$ is the per capita growth rate of the prey which we consider here to be logistic, i.e., $f(N)= r(1-\frac{N}{K})$ and $\mu$ is the intrinsic death rate of the predator which is assumed to be constant. Functional response of predator (the rate of food consumption per predator) which we consider here to be of Holling type II and we use the following parametrisation of $g(N)$ known as the Holling disk equation \cite{jost1999deterministic}
$$g(N)\,=\,\frac{aN}{1+aqN}.$$ 

In this model, we incorporate the Allee effect in the numerical response of the predator by assuming that its food conversion efficiency $e\psi(P)$ is a function of predator density. This is different from previous models where the Allee effect was also included in the functional response of the predator \cite{alves2017hunting,cosner1999effects,senallee}. The maximum food conversion coefficient is given by $e\, (0 < e<1)$ and this value is reached at high $P$. For simplicity, we neglect direct competitive effects and interference within the predator population. We assume that the reduction of the overall growth rate at high predator densities occurs solely due to over-exploitation of food, i.e., due to a decrease in $N$. At low predator density, the per capita reproduction rate becomes smaller, which is described by the function $\psi(P)$. Biologically, this can occur via a multitude of mechanisms. For example, due to low fertilization efficiency, lack of mating partners, sperm limitation, cooperative breeding mechanisms, etc \cite{berec2007multiple,courchamp2008allee}. Note this is a well-known approach to modelling the Allee effect by including the density dependence in the reproduction term and keeping the mortality term constant \cite{boukal2002single,morozov2016long}.


Mathematically, the function $\psi(P)$ expressing the Allee effect in the predator is considered to possess the following properties:
 \begin{enumerate}
\item[(A1)] $\psi(0) = 0$, since at very low densities the population cannot reproduce due to lack of mating opportunities;\\
\item[(A2)] $0 \leq \psi(P) \leq 1$ since $\psi(P)$ represents the proportion of the maximal possible conversion rate $e$; \\
\item[(A3)] $\psi(P)$ is an increasing function of $P$, so $\psi^{'}(P) > 0$ for all $P \geq 0$, and we do not include effects of intraspecific competition;\\
\item[(A4)] $\psi^{''}(P) < 0$ for all $P \geq 0$ which signifies that the increase in the reproductive ability (population fitness), while the population size $P$ is being increased, is monotonically decelerating;
\item[(A5)] $\psi(P) \rightarrow 1$ for large $P$ (we neglect intraspecific competition at high population sizes).
 \end{enumerate}

Next we reduce the number of parameters in the model by non-dimensionalisation and introduce the following non-dimensional variables $x\,=\,\frac{N}{K}$, $y\,=\,\frac{P}{krq}$ and $t\,=\,r \tau$ we can transform the equations (\ref{eq:1}) to
\begin{subequations} \label{eq:3}
\begin{align}
\begin{split}
\frac{dx}{dt}&=x\left(1-x\right)-\frac{xy}{\beta + x}\,\equiv\,F_{1}(x,y),
\end{split}\\
\begin{split}
\frac{dy}{dt}&=\frac{\alpha xy}{\beta + x}h(y) - my\,\equiv\,F_{2}(x,y),
\end{split}
\end{align}
\end{subequations}

\noindent with the following positive dimensionless parameters $\alpha\,=\,\frac{e}{qr}$, $\beta\,=\,\frac{1}{aqK}$
and $m\,=\,\frac{\mu}{r}$.

In the above model, the function $\psi(P)$ is transformed into a dimensionless function $h(y)$ with the same constraints as are imposed on $\psi(P)$. As we mentioned in the Introduction, we will explore the basic properties of the model for an arbitrary mathematical formulation of $h(y)$ (model equilibria, stability, possible generic bifurcation, etc), i.e., considering the above system as a partially specified model. We will also consider some concrete parameterisations of $h(y)$ such as $h(y)=\frac{y}{\delta + y}$ (Monod parametrisation), $h(y)=1-e^{-\frac{y}{\delta}}$ (Ivlev parametrisation) and $h(y)=\tanh(\frac{y}{\delta})$  (hyperbolic tangent parameterisation) to construct a full bifurcation portrait and explore the sensitivity of the model dynamics to mathematical formulation of the Allee effect. Finally, we verify how sensitive the model is with respect to small perturbations of $h(y)$ which still preserve assumptions (A1)-(A5).

\section{Model equilibria and their stability}

\subsection{Possible equilibria in the system}

We start our investigation by exploring the number and the location of system equilibria for an arbitrary formulation of the Allee effect $h(y)$. It is easy to see that model (\ref{eq:3}) always has one trivial equilibrium point $E_{0}\,=\,(0,0)$ and one axial equilibrium point $E_{1}\,=\,(1,0)$.

An interior equilibrium point $E^{*}\,=\,(x^{*},y^{*})$ will be a point of intersection of the following two non-trivial nullclines in the interior of first quadrant
\begin{subequations}\label{eq:nullcline}
\begin{align}
\begin{split}
f^{1}(x,y) &\equiv\, 1-x-\frac{y}{x+\beta}\,=\,0,
\end{split}\\[0.5em]
\begin{split}
f^{2}(x,y) &\equiv\, \frac{\alpha x}{x+\beta}h(y)-m\,=\,0.
\end{split}
\end{align}
\end{subequations}
For the feasibility of $y^{*}$, we must have $0 < x^{*} <1$ (see
equation (\ref{eq:nullcline}a)) and from this condition we can verify that $0<y^*\leq \frac{(1+\beta)^2}{4}$. We solve (\ref{eq:nullcline}b) for $x$ to obtain the equation for the predator nullcline
\begin{eqnarray}\label{eq:xfeasible}
x &=& \frac{m \beta}{\alpha h(y)-m}.
\end{eqnarray}
For the feasibility of $x^{*}$, we must have $y^{*} >
h^{-1}(\frac{m}{\alpha})$. Since the Allee effect function $h(y)$ is
bounded by 1, from $h(y)=\frac{m(x+\beta)}{\alpha x}$ we find
$x^{*}
> \frac{m \beta}{\alpha -m}$ with $\alpha>m$. We differentiate equation (\ref{eq:xfeasible}) with respect to $y$ to obtain
$$\frac{dx}{dy}\,=\,-\frac{\alpha m \beta h^{'}(y)}{(\alpha h(y)-m)^2} <0 \quad \textup{as} \quad h^{'}(y) > 0.$$
For the second derivative of the predator nullcline we have
$$\frac{d^{2}x}{dy^{2}}\,=\,\frac{\alpha \beta m}{(\alpha h(y)-m)^2}\left[\frac{2 \alpha (h^{'}(y))^2}{\alpha h(y)-m}-h^{''}(y)\right] > 0,$$
as $h^{''}(y) < 0$ and $\alpha h(y) > m$.
\begin{figure}[ht]
\centerline{
\includegraphics[scale=0.5]{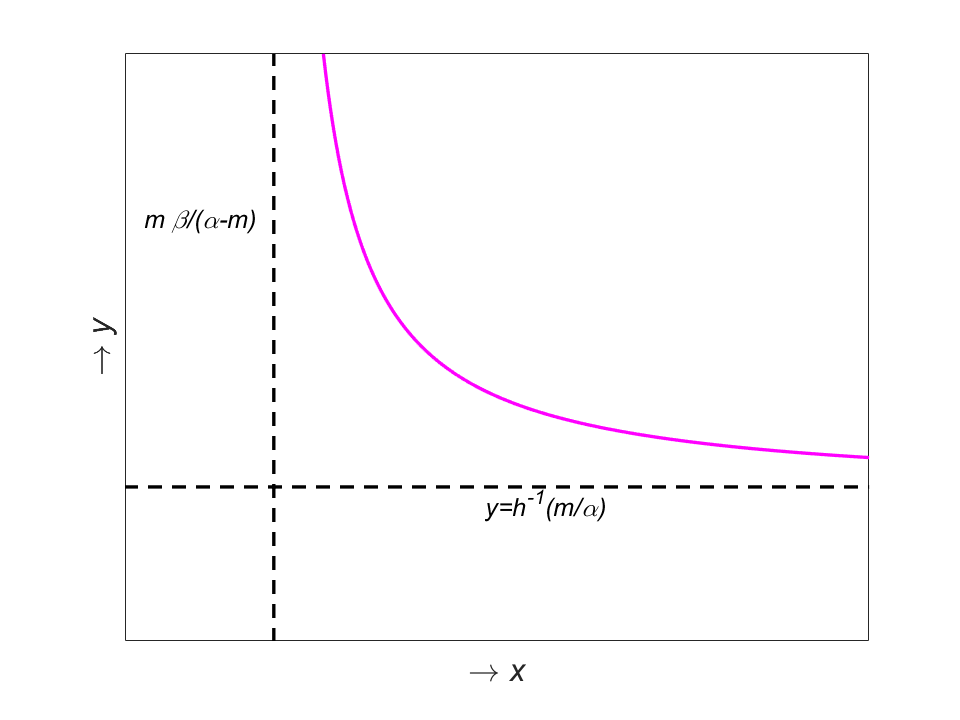}} \caption{Qualitative behviour of the predator nullcline in model (\ref{eq:3}).} \label{fig:predatorNC}
\end{figure}
\begin{figure}[ht]
\centerline{
\includegraphics[scale=0.5]{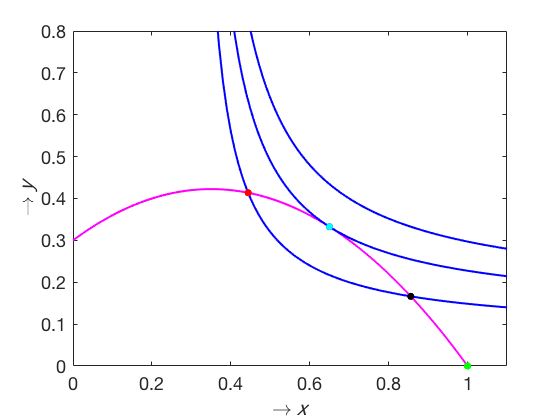}} \caption{Relative position of nullclines in model (\ref{eq:3}) constructed for the Allee effect parametrisation  given by the Monod function $h(y)=\frac{y}{\delta + y}$ with $\delta\,=\,0.4$ (no equilibrium), $\delta\,=\,0.315$ (a single equibrium), and $\delta\,=\,0.1$ (two equilibria points): one is a saddle point (black dot), the other is a topological focus (red dot). Other parameters are $\alpha\,=\,1.8$, $\beta\,=\,0.4$ and $m\,=\,0.5$} \label{fig:nullclines}
\end{figure}
From the above we derive that the predator nullcline  (\ref{eq:nullcline}b) is strictly decreasing as a function of $x$ and  it is always convex. Furthermore, the curve lies in the region where $x > \frac{m \beta}{\alpha -m}$ and $y
> h^{-1}(\frac{m}{\alpha})$. Taking into account the above properties, a possible shape of the predator nullcline is shown in
Fig.~\ref{fig:predatorNC}; the dashed lines represents the vertical and the horizontal asymptotes.


From the geometric properties of the nullclines one can see that there will be at most two points of intersection (between the non-trivial nullclines in the interior of the first quadrant) and so there can be at most two interior equilibria. An example of intersection of the model nullclines for the parameterisation of $h(y)$ given by $h(y)=\frac{y}{\delta + y}$ is shown in Fig.~\ref{fig:nullclines}. One can see that a gradual increase in $\delta$ (which defines characteristic predator densities at which the Allee effect has a pronounced strength) from small to large values results a saddle-node bifurcation which is described in detail in the next sections. Note that this property is observed for the other two parametrisations of $h(y)$ considered. Note that in the absence of the Allee effect, only one non-trivial equilibrium is possible, corresponding to the intersection of the vertical line $x=m\beta/(\alpha -m)$ and the prey nullcline.

\subsection{Stability of Equilibria}
Here  we explore the stability for the equilibria of model (\ref{eq:3}). The following proposition defines the stability of the axial equilibria.

\begin{prop}
For any choice of $h(y)$ satisfying assumptions (A1)-(A5)
\begin{itemize}
\item[(i)] the trivial equilibrium point $E_{0}$ is a saddle;\\
\item[(ii)] the axial equilibrium point $E_{1}$ is locally asymptotically stable.
\end{itemize}
\end{prop}

\begin{proof}
The Jacobian matrix of model (\ref{eq:3}) at any point is given by

\begin{eqnarray}
\label{eq:jacobian}
J(x,y) &=& \begin{bmatrix}
1-2x-\frac{y}{\beta + x}+\frac{xy}{(\beta + x)^2} & -\frac{x}{\beta + x}\\[1em]
\frac{\alpha \beta y h(y)}{(\beta + x)^2} & \frac{\alpha x}{\beta + x}h(y) + \frac{\alpha x y}{\beta + x}h^{'}(y)-m
\end{bmatrix}.
\end{eqnarray}

\begin{itemize}
\item[(i)] The eigenvalues of the Jacobian matrix at $E_{0}$ are $1$ and $-m$. Therefore it is a saddle point irrespective of the choice of $h(y)$, having a stable manifold along $y$-axis and an unstable manifold along $x$-axis.

\item[(ii)] The axial equilibrium point $E_{1}$ is locally asymptotically stable (a stable node) as the eigenvalues of the Jacobian matrix are $-1$ and $-m$ for any choice of $h(y)$.
\end{itemize}
\end{proof}

An important conclusion is that, in the presence of an Allee effect in the predator, achieving a very low population densities by the predator will result in its eventual extinction, so the Allee effect is strong.

Next we explore the stability of the interior equilibria. As follows from the previous section, model (\ref{eq:3}) admits at most two interior equilibrium points which we  denote by $E_{1*}(x_{1*},y_{1*})$ and $E_{2*}(x_{2*},y_{2*})$ such that
$0\,<\,x_{1*}\,<\,x_{2*}\,<\,1$. The Jacobian matrix evaluated at $E_{*}\,=\,(x_{*},y_{*})$ can be expressed as
\begin{eqnarray}\label{eq:jacobian}
J(E^{*}) &= \begin{bmatrix}
xf^{1}_{x} & xf^{1}_{y}\\
yf^{2}_{x} & yf^{2}_{y}
\end{bmatrix}_{(x_{*},y_{*})},
\end{eqnarray}
where we have used $f^{1}(x_{*},y_{*})\,=\,0$ and $f^{2}(x_{*},y_{*})\,=\,0$.  Since $f^{1}$ and $f^{2}$ are smooth functions, we can differentiate both expressions (\ref{eq:nullcline}) to obtain
\begin{eqnarray*}
f^{1}_{x} &=& -f^{1}_{y}\frac{dy^{(f^{1})}}{dx}, f^{2}_{x} = -f^{2}_{y}\frac{dy^{(f^{2})}}{dx},
\end{eqnarray*}
where $\frac{dy^{(f^{1})}}{dx}$ and $\frac{dy^{(f^{2})}}{dx}$ are tangent lines to the nullclines $f^{1}(x,y)\,=\,0$ and $f^{2}(x,y)\,=\,0$, respectively. We substitute the above expressions into the Jacobian matrix
\begin{eqnarray}
J(E^{*}) &= \begin{bmatrix}
-xf^{1}_{y} \frac{dy^{(f^{1})}}{dx}  & xf^{1}_{y}\\
-yf^{2}_{y} \frac{dy^{(f^{2})}}{dx} & yf^{2}_{y}
\end{bmatrix}_{(x_{*},y_{*})}.
\end{eqnarray}
Therefore, for the determinant of the Jacobian we obtain
\begin{eqnarray}\label{eq:det}
\textup{Det}(J(E^{*}_{i})) &= \left[xyf^{1}_{y}f^{2}_{y}\left(\frac{dy^{(f^{2})}}{dx}-\frac{dy^{(f^{1})}}{dx}\right)\right]_{(x_{*},y_{*})}.
\end{eqnarray}
Now $f^{1}_{y}(x,y)\,=\,-\frac{1}{\beta + x} < 0$ and $f^{2}_{y}(x,y)\,=\,\frac{\alpha x}{\beta + x}h^{'}(y) > 0$ as $h^{'}(y) > 0$. Substituting the above derivatives we have
\begin{subequations}\label{eq:fracderi}
\begin{align}
\begin{split}
\frac{dy^{(f^{1})}}{dx} &= 1-\beta -2x \,=\,2(x_{m}-x)\end{split},\\
\begin{split}
\frac{dy^{(f^{2})}}{dx} &= -\frac{\beta h(y)}{\beta + x}\frac{1}{xh^{'}(y)} < 0 ,
\end{split}
\end{align}
\end{subequations}
where $x_{m}\,=\,\frac{1-\beta}{2}$ is the $x$-coordinate of the point where the nullcline $f^{1}(x,y)\,=\,0$ attains its maximum within the first quadrant.

Depending upon the positions of two points of intersections between the two nontrivial nullclines (cf. (\ref{eq:nullcline})) with respect to the point of maximum on the prey nullcline, we can consider following two cases,

\noindent \textbf{Case:1}\, $0 < x_{m} < x_{1*} < x_{2*} < 1$ and \noindent \textbf{Case:2}\, $0 < x_{1*} < x_{m} < x_{2*} < 1$.

\begin{figure}[ht]
\centerline{
\includegraphics[scale=0.5]{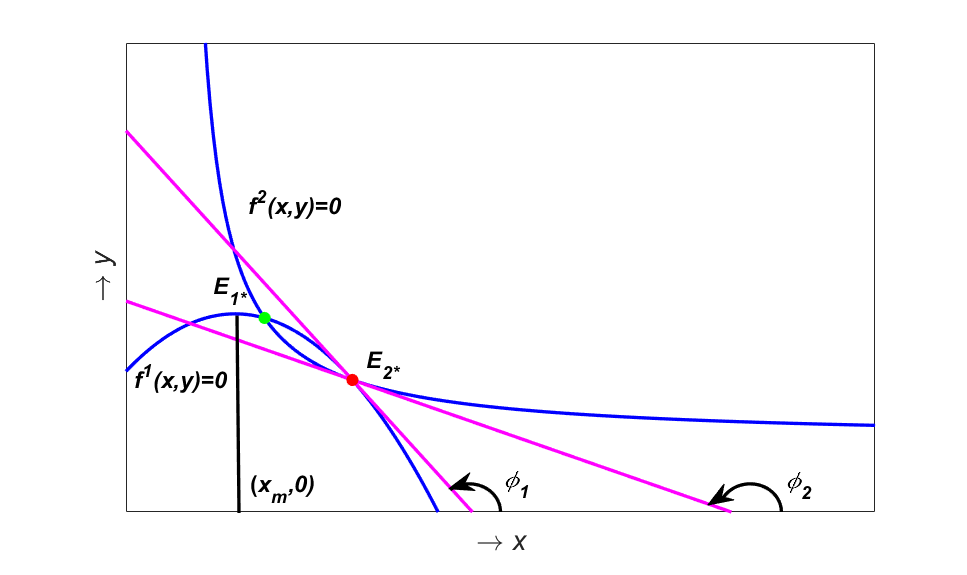}} \caption{ Determing the type and stability of interior equilibria in model (\ref{eq:3}). For detail see the text.} \label{fig:Nullcline_Tangent1}
\end{figure}

\noindent \textbf{Case:1}\, Suppose $\phi_1$ and $\phi_2$
are angles made by the tangents to $f^1(x,y)\,=\,0$ and
$f^2(x,y)\,=\,0$ at $E_{2*}$. Then from Fig.~\ref{fig:Nullcline_Tangent1} one can see that $\frac{\pi}{2}\,<\,\phi_1\,<\,\phi_2\,<\,\pi$ which implies
\begin{eqnarray}
\left.\frac{dy^{(f^2)}}{dx}\right|_{E_{2*}}\,>\,\left.\frac{dy^{(f^1)}}{dx}\right|_{E_{2*}}.
\end{eqnarray}
Therefore from the above inequality and expression (\ref{eq:det}) we get $\textup{Det}(J(E_{2*})) < 0$. Hence $E_{2*}$ is a saddle  point.

We can proceed in a similar fashion and use the fact that (this is not shown in Fig. ~\ref{fig:Nullcline_Tangent1} for brevity)
\begin{eqnarray}
\left.\frac{dy^{(f^1)}}{dx}\right|_{E_{1*}}\,>\,\left.\frac{dy^{(f^2)}}{dx}\right|_{E_{1*}}
\end{eqnarray}
and to prove that $\textrm{Det}\left(J(E_{1*})\right)\,>\,0$, i.e., $E_{1*}$ is not a saddle point. The stability of $E_{1*}$ is determined by the sign of
$f_x^1(x_{1*},y_{1*})+f_y^2(x_{1*},y_{1*})$. Note that
\begin{subequations}
\begin{eqnarray}
f_x^1(x_{1*},y_{1*})&=&2(x_{m}-x_{1*}) < 0 \quad \textup{as} \quad x_{m} < x_{1*},\\
f_y^2(x_{1*},y_{1*})&=&\frac{\alpha x_{1*}h'(y_{1*})}{\beta + x_{1*}}\,>\,0.
\end{eqnarray}
\end{subequations}
Hence $E_{1*}$ is locally asymptotically stable if
\begin{eqnarray}
f_x^1(x_{1*},y_{1*})+f_y^2(x_{1*},y_{1*})\,<\,0.
\end{eqnarray}

\textbf{Case:2} In this case we can also prove that $E_{2*}$ is a saddle point proceeding in a similar manner as above. For the stability of $E_{1*}$ we have from (\ref{eq:fracderi}a), $\left. \frac{dy^{(f^1)}}{dx}\right|_{E_{1*}} > 0$. Hence we get from (\ref{eq:det}), $\textup{Det}(J(E_{1*})) > 0$. Also
\begin{subequations}
\begin{eqnarray}
f_x^1(x_{1*},y_{1*})&=&2(x_{m}-x_{1*}) > 0 \quad \textup{as} \quad x_{m} > x_{1*},\\
f_y^2(x_{1*},y_{1*})&=&\frac{\alpha x_{1*}h'(y_{1*})}{\beta + x_{1*}}\,>\,0.
\end{eqnarray}
\end{subequations}
As $\textup{tr}(J(E_{1*}))\,=\,f_x^1(x_{1*},y_{1*}) + f_y^2(x_{1*},y_{1*}) > 0$, hence $E_{1*}$ is unstable.

To conclude, the interior equilibrium $E_{2*}$ is always a saddle point, whereas $E_{1*}$ is a topological focus which depending on parameters can be either stable or unstable.

\section{Local bifurcations in the model}

Here we consider possible local bifurcations in model (\ref{eq:3}).

\subsection{Saddle-node bifurcation}

Suppose $\overline{E}(\overline{x},\,\overline{y})$ is the point at which two non-trivial nullclines touch each other in the first quadrant when a bifurcation parameter of the model is being varied.  The slope of the tangents to the curves at $\overline{E}$ become equal. This signifies that
\begin{eqnarray}
\left.\frac{dy^{(f^1)}}{dx}\right|_{(\overline{x},\,\overline{y})}\,=\,
-\left.\frac{f_x^1}{f_y^1}\right|_{(\overline{x},\,\overline{y})}\,=\,
-\left.\frac{f_x^2}{f_y^2}\right|_{(\overline{x},\,\overline{y})}\,=\,
\left.\frac{dy^{(f^2)}}{dx}\right|_{(\overline{x},\,\overline{y})},
\end{eqnarray}
and hence we have
\begin{eqnarray}
\left.f_x^1f_y^2-f_y^1f_x^2\right|_{(\overline{x},\,\overline{y})}\,=\,0.
\end{eqnarray}
In this case $\textrm{Det}\left(J(\overline{E})\right)\,=\,0$ and
$\overline{E}$ becomes a non-hyperbolic equilibrium point. This situation corresponds to a saddle-node bifurcation in the model.  We explore this bifurcation in more detail.

As an example, we consider $m$ to be the bifurcation parameter and denote by
$m\,\equiv\,m_{SN}$ the bifurcation  point. The eigenvectors of both the matrix $J(\bar{E})$ and its transpose corresponding to the zero
eigenvalue are, respectively given by  $v\,=\,[1,1-\beta-2\bar{x}]^{t}$ and $w\,=\,[\alpha
\bar{y}h^{'}(\bar{y}),1]^{t}$. We need to check the transversality conditions for a saddle-node bifurcation \cite {perko2013differential}. We denote
$F\,=\,(F_{1}(x,y),F_{2}(x,y))^{t}$ and we further follow the same notation of \cite {perko2013differential} to obtain
\begin{eqnarray*}
w^{t}F_{m}(\bar{E};m=m_{SN}) &=& -1\,\neq\,0,\\
w^{t}D^{2}F(\bar{E};m=m_{SN})(v,v) &=& -\frac{2\alpha \bar{x}\bar{y}h^{'}(\bar{y})}{\beta + \bar{x}}-\frac{\beta \bar{y}h(\bar{y})}{(\beta + \bar{x})^2}\\
 & & \left[2 + \frac{2\beta}{\bar{x}(\beta + \bar{x})} - \frac{\beta h(\bar{y})h^{''}(\bar{y})}{\bar{x}(\beta + \bar{x})(h^{'}(\bar{y}))^2}\right] \,<\,0,
\end{eqnarray*}
as $h^{''}(y) < 0$. Hence the transversality conditions are always satisfied and variation of $m$  results in a saddle-node bifurcation. Similar results can be obtained by varying other model parameters.

\subsection{Hopf Bifurcation}

In the previous subsection we show that the two interior equilibrium points are generated through a saddle-node bifurcation. The non-saddle interior equilibrium ($E_{1*}$) can be stable or unstable depending on model parameters. It loses its stability when the sign of the trace of the Jacobian matrix has changed through zero (from negative to positive) via a Hopf bifurcation. In this section we show that system (\ref{eq:3}) undergoes a Hopf bifurcation when a model parameter is varied. Here we choose $\beta$ as a bifurcation parameter. The Jacobian matrix at $E_{1*}$ is given by

\begin{eqnarray*}
J(E_{1*}) &=& \begin{bmatrix}
\frac{x_{1*}}{\beta+x_{1*}}(1-\beta-2x_{1*}) & -\frac{x_{1*}}{\beta+x_{1*}}\\[0.5em]
\frac{\alpha \beta y_{1*}h(y_{1*})}{(\beta+x_{1*})^2} & \frac{\alpha x_{1*}y_{1*}}{\beta+x_{1*}}h^{'}(y_{1*}).
\end{bmatrix}
\end{eqnarray*}

\noindent Now let assume that $\beta\,=\,1-2x_{1*}+\alpha y_{1*}h^{'}(y_{1*})\,\equiv\,\beta_{H}$. This is an implicit expression for $\beta$ as the components of the equilibrium point contain $\beta$ as well. Now we assume that the following three conditions are satisfied at $\beta\,=\beta_H$,

\begin{itemize}
\item[(H1)] $T_{H}\,=\,\textup{tr}(J(E_{1*};\beta\,=\,\beta_{H}))\,=\,0$,\\
\item[(H2)] $\Delta_{H}\,=\,\textup{det}(J(E_{1*});\beta\,=\,\beta_{H}) > 0$,\\
\item[(H3)] If $\lambda(\beta)$ is the complex eigenvalue of $J(E_{1*})$ then $\left.\frac{d}{d\beta}\left(\textup{Re}(\lambda(\beta))\right)\right|_{\beta=\beta_H}\,\neq\,0.$
\end{itemize}

Then $E_{1*}$ loses its stability through a Hopf bifurcation at $\beta\,=\,\beta_{H}$.\\

Assuming $\textup{Re}(\lambda(\beta))$ is real part of a complex eigenvalue of $J(E_{1*})$, we can write,
$$\textup{Re}(\lambda(\beta))=tr(J(E_{1*}))/2\,=\,\frac{x_{1*}}{2(\beta+x_{1*})}\left\lbrace1-\beta-2x_{1*}+\alpha y_{1*}h^{'}(y_{1*})\right\rbrace.$$
\noindent Now $\textup{Re}(\lambda(\beta))$ is equal to zero when $\beta\,=\,\beta_H$. Differentiating $\textup{Re}(\lambda(\beta))$ with respect to $\beta$ we find that
\begin{eqnarray}
\label{eq:eigen_thresh}
\frac{d}{d\beta}\{\textup{Re}(\lambda(\beta))\}&=&\frac{x_{1*}}{2(\beta + x_{1*})}\left\lbrace -1-2\frac{dx_{1*}}{d\beta}+\alpha\frac{dy_{1*}}{d\beta}\left(h^{'}(y_{1*})+y_{1*}h^{''}(y_{1*})\right)\right\rbrace\nonumber\\
& & +\frac{\beta}{2(\beta+x_{1*})^2}\frac{dx_{1*}}{d\beta}\left\lbrace1-\beta-2x_{1*}+\alpha y_{1*}h^{'}(y_{1*})\right\rbrace.
\end{eqnarray}
\noindent Now as $E_{1*}$ satisfies (\ref{eq:nullcline}a) we have,
$$\frac{dy_{1*}}{d\beta}\,=\,1-x_{1*}+\frac{dx_{1*}}{d\beta}\left(1-2x_{1*}-\beta\right).$$
\noindent Finally using the fact $1-\beta-2x_{1*}+\alpha y_{1*}h^{'}(y_{1*})\,=\,0$ at $\beta\,=\,\beta_H$ and above result in (\ref{eq:eigen_thresh}) we get

\begin{eqnarray*}
\left.\frac{d}{d\beta}\{\textup{Re}(\lambda(\beta))\}\right|_{\beta\,=\,\beta_{H}} &=& \frac{x_{1*}}{2(\beta_{H} + x_{1*})}\left[-1+\alpha(1-x_{1*})\left(h^{'}(y_{1*})+y_{1*}h^{''}(y_{1*})\right)-\right.\\
 & & \left.\frac{dx_{1*}}{d\beta}\left(2-\alpha (1-2x_{1*}-\beta_{H})(h^{'}(y_{1*})+y_{1*}h^{''}(y_{1*}))\right)\right]_{\beta=\beta_H}.
\end{eqnarray*}

The above expression should be checked for the given mathematical formulation of the Allee effect $h(y)$. In particular, we have numerically verified that this quantity is non-zero at the Hopf bifurcation threshold for the parameterisations considered here: the Monod, Ivlev and trigonometric functions.

\subsection{Generalized Hopf (Bautin) bifurcation}

\noindent In this section we consider a co-dimension two bifurcation called a Bautin or generalized Hopf (GH) bifurcation. This bifurcation occurs when the interior non-saddle equilibrium has purely imaginary eigenvalues and the first Liapunov number becomes zero. We consider $\beta$ and $m$ as bifurcation parameters. Therefore in $\beta$-$m$ parametric plane, there is a critical point which lies on the Hopf bifurcation curve. In the next proposition we will show that the model undergoes a GH bifurcation by choosing $\beta$ and $m$ as bifurcation parameter.

\begin{prop}
Model (\ref{eq:3}) undergoes a Bautin  (generalized Hopf) bifurcation around the interior equilibrium point $\hat{E}\,=\,(\hat{x},\hat{y})$ at the bifurcation threshold $(\beta_{GH},m_{GH})$ whenever the following conditions hold
\begin{itemize}
\item[(GH1)] $\Delta_{GH}\,=\,\textup{det}(J(\hat{E});\beta\,=\,\beta_{GH},\,m\,=\,m_{GH}) > 0$,\\
\item[(GH2)] $T_{GH}\,=\,\textup{tr}(J(\hat{E});\beta\,=\,\beta_{GH},\,m\,=\,m_{GH})\,=\,0$,\\
\item[(GH3)] $l(\hat{E};\beta\,=\,\beta_{GH},\,m\,=\,m_{GH}))\,=\,0$,
\end{itemize}
\noindent where $l$ is the first Liapunov number.
\end{prop}

\begin{proof}
See supplementary material SM1.
\end{proof}

Examples of the above type of bifurcation for several parameterisations of $h(y)$ are provided in Section 5.

\subsection{Bogdanov-Takens bifurcation}

Another type of co-dimension two local bifurcation observed in model (\ref{eq:3}) is a Bogdanov- Takens (BT) bifurcation. In a two dimensional parametric plane, this bifurcation occurs at a point where a Hopf bifurcation curve meets a saddle-node bifurcation curve tangentially. In the previous section, we chose $\beta$ as the bifurcation parameter for the Hopf bifurcation and $m$ for the saddle-node bifurcation. Therefore we will consider $\beta$ and $m$ as bifurcation parameters for the BT bifurcation and suppose that model (\ref{eq:3}) exhibits a BT bifurcation at $\bar{E}\,=\,(\bar{x},\bar{y})$ and the parametric thresholds are denoted by $(\beta,m)\,=\,(\beta_{BT},m_{BT})$. From the general bifurcation theory \cite{perko2013differential} it is known that $\bar{E}$  satisfies the equations of nullclines (\ref{eq:nullcline}) and also the  Jacobian matrix is similar to $\begin{bmatrix}
0 & 1\\
0 & 0
\end{bmatrix}$ at $\bar{E}$ for the parameter threshold $(\beta\,=\,\beta_{BT},m\,=\,m_{BT})$. The following proposition provides the conditions for model (\ref{eq:3}) to  undergo a Bogdanov-Takens bifurcation.

\begin{prop}
If we choose $\beta$ and $m$ as bifurcation parameters, then system (\ref{eq:3}) undergoes a Bogdanov-Takens bifurcation around the interior equilibrium point $\bar{E}$ whenever the following conditions hold
\begin{itemize}
\item[(BT1)] $\textup{tr}(J(\bar{E};\beta\,=\,\beta_{BT},\,m\,=\,m_{BT}))\,=\,0$ ,\\
\item[(BT2)] $\textup{det}(J(\bar{E};\beta\,=\,\beta_{BT},\,m\,=\,m_{BT}))\,=\,0$.\\
\end{itemize}
\end{prop}

\begin{proof}
See supplementary material SM 2 for detail.
\end{proof}

We found that model (\ref{eq:3}) may undergo a Bogdanov-Takens bifurcation of either co-dimension 2 or co-dimension 3. The latter requires an extra condition of degeneracy given in the supplementary material. Note that a co-dimension 3 Bogdanov-Takens bifurcation, if it exists, should be of the type involving a double equilibrium point \cite{dumortier1987generic}. Indeed, the other type of this bifurcation - known as the cusp- would require a triple equilibrium point which is impossible for this model as shown in Section 3. For the same reason, a co-dimension 4 Bogdanov-Takens  bifurcation is impossible in this system. Examples of Bogdanov-Takens bifurcation of co-dimension 2 and 3 for particular parameterisations of $h(y)$ are provided in the next section.


\section{Parametric diagrams and phase portraits}
In this section, we construct global parametric diagrams for the considered model for  three different mathematical formulations of the Allee effect $h(y)$ given by the Monod, Ivlev and trigonometric tangent  functions. Note that all of them satisfy assumptions (A1)-(A5). Examples of all three curves constructed for $\delta=0.2$ are shown in Fig. ~\ref{fig:Allee}. Note that for the plotted functions the initial slopes and their asymptotic values for large $y$ are the same.
\begin{figure}[ht]
\centerline{
\includegraphics[scale=0.5]{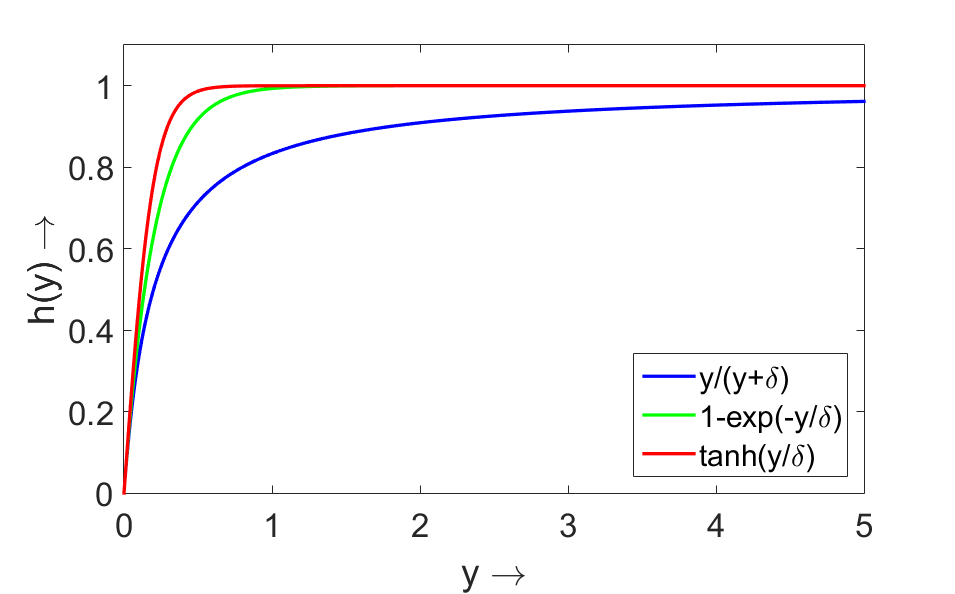}} \caption{Graphs of parameterisations of the Allee effect in the predator given by the Monod ($y/(y+\delta)$), the Ivlev ($1-\exp(-y/ \delta)$) and the trigonometric ($\tanh(y/ \delta)$) functions; $\delta=0.2$.} \label{fig:Allee}
\end{figure}

The model contains 4 parameters, so it is convenient to present our results in a 3 dimensional parametric space and then  explore the alteration to the portrait by varying a fourth parameter. We construct portraits in the $(\alpha,\delta,m)$ space with a further variation of $\beta$. For all considered formulations of $h(y)$, the parameter $\delta$ can be interpreted as the intensity of the Allee effect. In particular, in the case where $\delta$ vanishes the system becomes the classical Rosenzweig-MacArthur predator-prey model.

\begin{figure}
\centering
\mbox{\subfigure[]{\includegraphics[width=6.3cm,height=5.5cm]{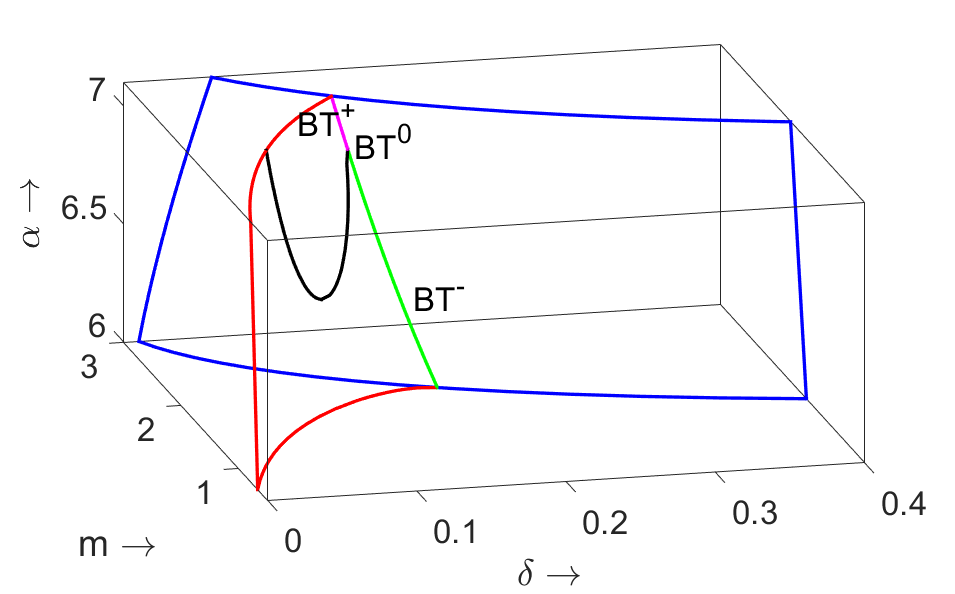}}
\subfigure[]{\includegraphics[width=6.3cm,height=5.5cm]{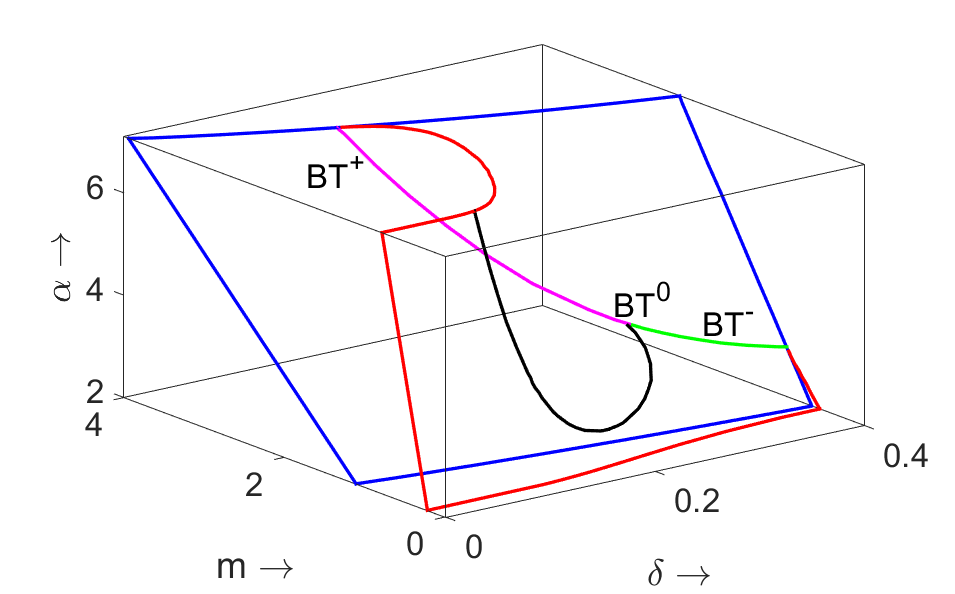}}}
\mbox{\subfigure[]{\includegraphics[width=7cm,height=6cm]{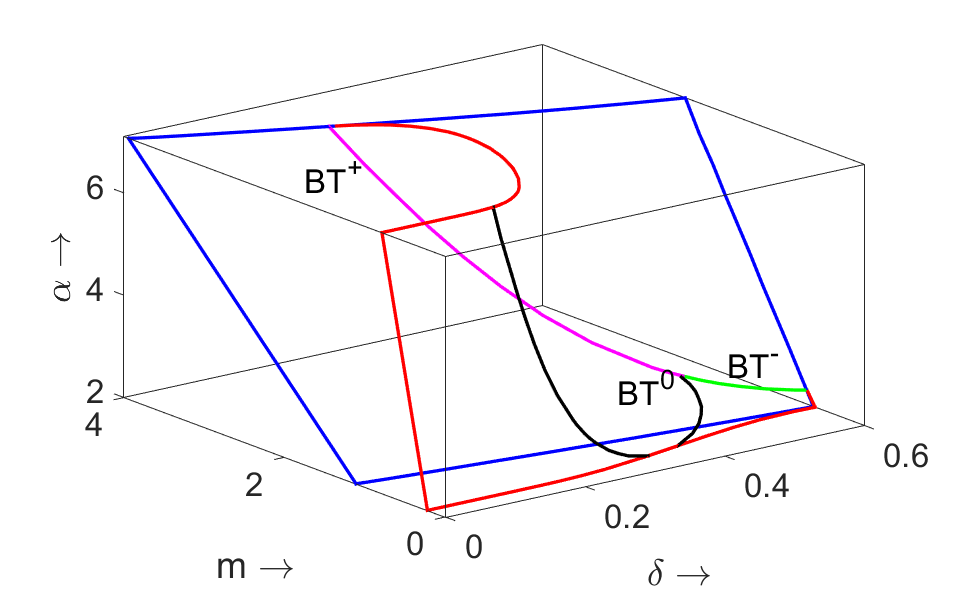}}}
\caption{Three dimensional diagram $(\alpha,\delta,m)$ of model (\ref{eq:3}) for the Allee effect parameterised by: (a) the Monod response $h(y)=y/(\delta+y)$; (b) the Ivlev response $h(y)=1-e^{-\frac{y}{\delta}}$; (c) the trigonometric tangent function $h(y)=\tanh(\frac{y}{\delta})$. In each diagram, $\beta=0.8$. The explanation of the surfaces and curves is in the text.}
\label{Fig_3D_diagrams}
\end{figure}

Examples of parametric portraits for the three functional forms of $h(y)$ are given in Fig. \ref{Fig_3D_diagrams}, in each case  $\beta=0.8$ is kept fixed. We show the skeletons of the parametric portraits given by local bifurcations: to avoid overloading the diagram, we do not include the non-local bifurcations which are shown in the corresponding cross sections in next figures. The saddle-node bifurcation surface is denoted by the blue curves. These curves show intersections of the saddle-node bifurcation surface with boundaries of the parameteric diagram. The intersection of the Hopf bifurcation surface with the boundaries is denoted by red curves. The saddle-node and Hopf surfaces intersect along the Bogdanov - Takens bifurcation curve which consists of green and magenta coloured parts: the magenta colour corresponds to a Bogdanov- Takens bifurcation of codimension 2 with a positive product of the state variables in the normal form (see SM 2 for detail) and is denoted as $BT^+$; the green part of the curve gives Bogdanov-Takens bifurcation of codimension 2 with this product having negative sign and is denoted as $BT^-$. The black curve represents the location of Generalised Hopf points on the Hopf bifurcation surface. This curve emerges from the point of Bogdanov-Takens bifurcation of codimension 3 (denoted as $BT^0$). From comparison of the diagrams in Fig. \ref{Fig_3D_diagrams}, we conclude that the global bifurcation structure in the parametric space remains the same topologically for all three formulations of $h(y)$.


To better understand the parametric structure and feasible phase portraits in the model, we explored two-dimensional cross sectional diagrams for a constant $\alpha$ and $\beta$. In the main text, we present the diagrams for the Monod formulation of $h(y)$. The diagrams for the other functional forms of $h(y)$ are shown in the supplementary material (SM 3). An example of a $(\delta,m)$ diagram constructed for $\alpha$ above the $BT^0$ point is shown in Fig. ~\ref{Fig_bifdiagSaturating}a; the other two $(\delta,m)$ diagrams in the same figure are constructed for $\alpha$ below the $BT^0$ point. Fig.\ref{Fig_bifdiagSaturating}b describes the situation where the  $(\delta,m)$ plane does not intersect the GH bifurcation curve, the opposite case is shown in  Fig.\ref{Fig_bifdiagSaturating}c. The corresponding phase portraits of the model are given in Fig. \ref{phaseportraitsHolling}.

\begin{figure}
\centering
\mbox{\subfigure[]{\includegraphics[width=6.3cm,height=5.5cm]{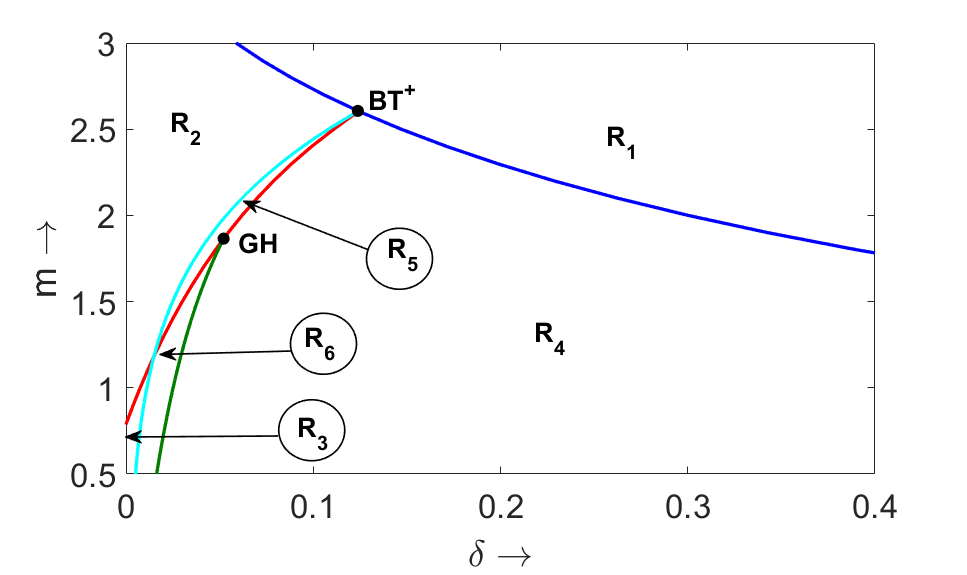}}
\subfigure[]{\includegraphics[width=6.3cm,height=5.5cm]{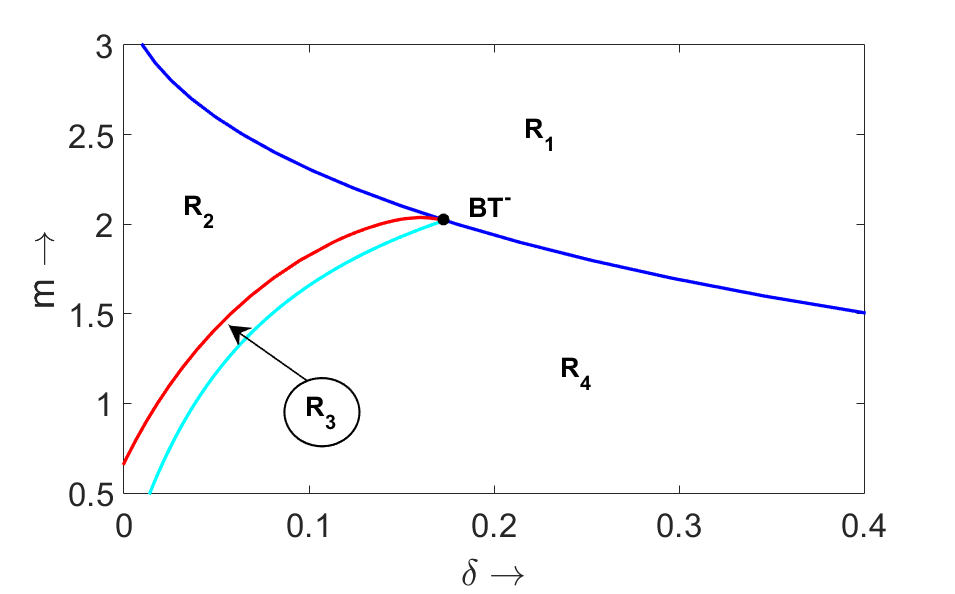}}}
\mbox{\subfigure[]{\includegraphics[width=6.3cm,height=5.5cm]{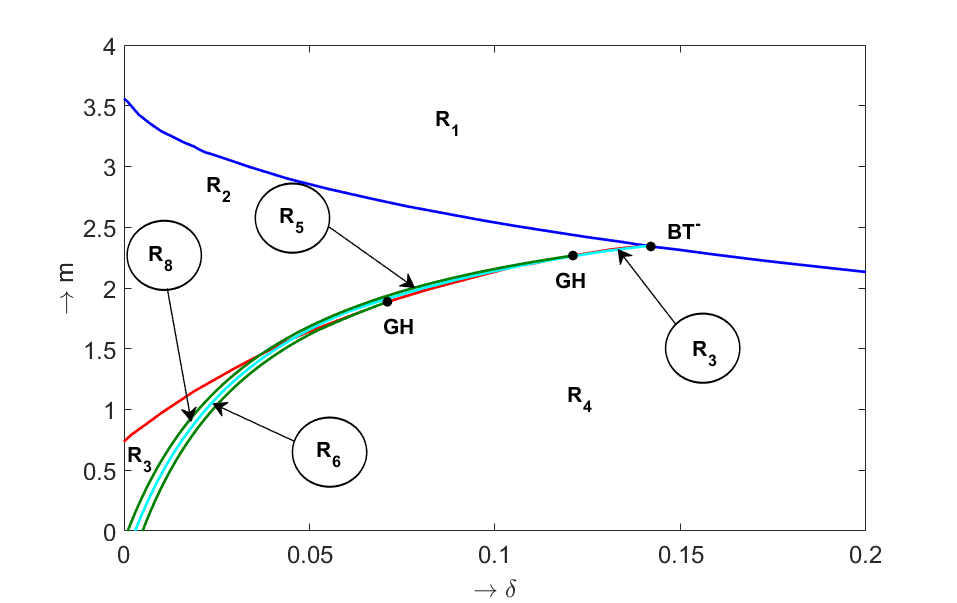}}
\subfigure[]{\includegraphics[width=6.3cm,height=5.5cm]{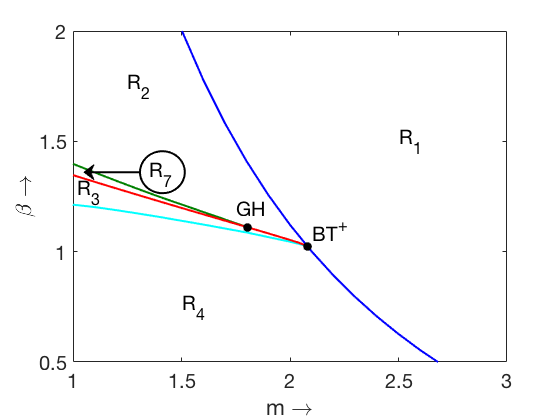}}}
\caption{(a)-(c) Examples of bifurcation diagrams in $\delta$-$m$ parametric
plane for the Monod parameterisation of the Allee effect in predator
plotted for $\alpha$ above and below the $BT^0$ point shown in
Fig.\ref{Fig_3D_diagrams} (a). The parameter values are: (a)
$\alpha\,=\,7.1$; (b) $\alpha\,=\,6.0$ and (c)
$\alpha\,=\,6.6$; in all cases $\beta\,=\,0.8$. 
(d) Two dimensional diagram $(\beta, m)$ constructed for $ \alpha=7.1$ and $\delta=0.2$. In each diagram, the dark blue curve is a saddle-node bifurcation curve. The red curve is
a Hopf bifurcation curve. The green curve is the curve of a fold
bifurcation of limit cycles. The cyan curve describes a homoclinic
bifurcation. The meaning of the regions $R_i$ is explained in
Fig.\ref{phaseportraitsHolling} and in the main text.}
\label{Fig_bifdiagSaturating}
\end{figure}

From Fig. \ref{Fig_bifdiagSaturating} (a)-(c) one can see that for large values of $m$ and $\delta$ (region $R_1$) there are no coexistence equilibria in the system: the only (global) attractor is the state $E_{1}$, where only the prey survives, whereas the predator goes to extinction. The corresponding phase portrait is shown in Fig. \ref{phaseportraitsHolling}(a).

Reduction in the strength of the Allee effect (small $\delta$ and high rates of mortality $m$)	results in the emergence of a pair of equilibrium points: a saddle $E_{1*}$ and a node   $E_{2*}$ (region $R_2$). The non-saddle point is only locally stable: its basin of attraction is limited by that of the axial equilibrium point $E_{1}$, which is shown in Fig. \ref{phaseportraitsHolling}(b). For large values of $m$ (on the right hand side of a BT point), a decrease in $\delta$ will result in a saddle-node bifurcation where the non-saddle point $E_{1*}$  will be unstable (region $R_4$). In this case, the global attractor will be the prey only state $E_{1}$ (the phase portrait is shown in Fig. \ref{phaseportraitsHolling}(d)). The loss of stability of $E_{1*}$  when crossing the Hopf bifurcation curve around the BT point depends on the sign of the BT point. For $BT^+$ (Fig. \ref{Fig_bifdiagSaturating}(a)) transition from $R_2$ to $R_4$ occurs via region $R_5$ by crossing the homoclinic loop bifurcation curve. A locally stable interior equilibrium $E_{1*}$ bcomes surrounded by an unstable cycle which forms its basin of attraction (Fig. \ref{phaseportraitsHolling}(e)). All trajectories starting outside this cycle will be attracted to the prey only state $E_1$. The transition from region $R_5$ to region $R_2$ occurs via a homoclinic loop bifurcation. For $BT^-$ (Fig. \ref{Fig_bifdiagSaturating}b), the transition from $R_2$ to $R_4$ occurs via region $R_3$ by crossing a supercritical Hopf bifurcation curve. In region $R_3$, an unstable internal equilibrium $E_{1*}$ is surrounded by a stable limit cycle (Fig. \ref{phaseportraitsHolling}(c)). One can see that in Fig. \ref{Fig_bifdiagSaturating}(a), for smaller $m$, a decrease in $\delta$ from region $R_4$ results in a fold bifurcation of limit cycles. In region $R_6$ we have two limit cycles: the inner cycle is stable, the outer cycle is unstable (Fig.\ref{phaseportraitsHolling}(f)). The outer cycle forms the boundary of the basin of attraction for the state $E_1$.

The diagram in Fig.\ref{Fig_bifdiagSaturating}c is more complicated
as compared to Fig.\ref{Fig_bifdiagSaturating}a,b. In particular, a
new region $R_8$ emerges, where three limit cycles can coexist: the
inner limit cycle is stable, the middle cycle is unstable and the outer
one is stable. The corresponding portrait is shown
in Fig.\ref{phaseportraitsHolling}(h).


For the Ivlev and trigonometric formulations of $h(y)$, the bifurcation diagrams in the $\delta$-$m$ plane constructed for $\alpha$ above and below the $BT^0$ point are topologically equivalent (see supplementary material SM3). However, for fixed $\beta$ and $\alpha$ the location of the bifurcation curves as well as the types of bifurcation (e.g., $BT^+$ versus $BT^-$ type of bifurcation) in the $\delta$-$m$ plane may be substantially different, especially when comparing the Monod parametrisation with the other two functional forms. This indicates sensitivity of the model to the functional form of $h(y)$. We explore the structural sensitivity in more detail in the next section.

Consider now variation of the fourth model parameter $\beta$. A decrease in $\beta$ results in a shift of the saddle-node and Hopf bifurcation surfaces in the $(\alpha,\delta,m)$ space closer to the $\alpha-\delta$ plane. The length of the curve of Bogdanov-Takens points (the intersection between the saddle-node and Hopf bifurcation) will be shortened and it moves upwards on the saddle-node bifurcation surface. The codimension 3 Bogdanov-Takens bifurcation is still  observed. The above properties hold true for all three functional forms of $h(y)$ considered.


\begin{figure}
\centering
\mbox{\subfigure[]{\includegraphics[width=0.33\textwidth]{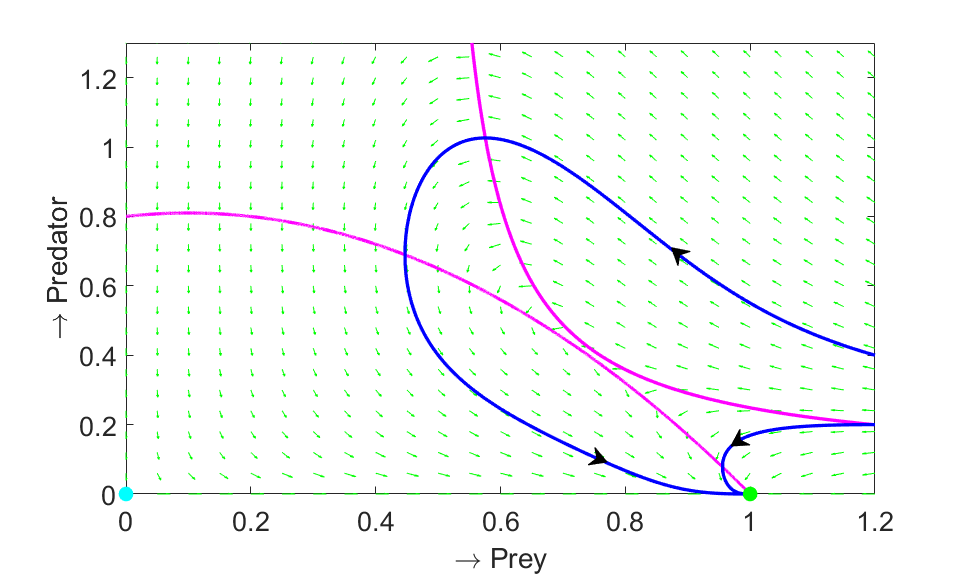}}
\subfigure[]{\includegraphics[width=0.33\textwidth]{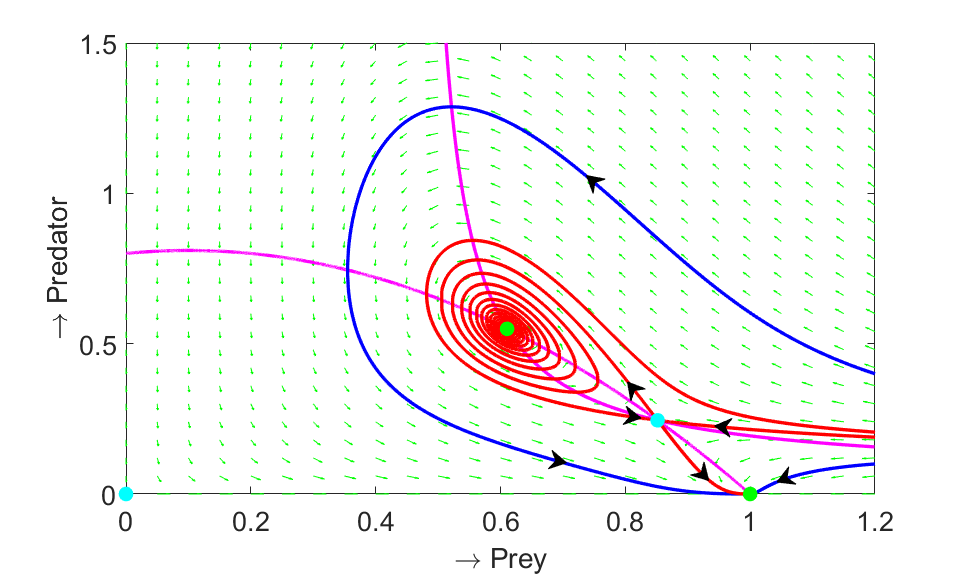}}}
\mbox{\subfigure[]{\includegraphics[width=0.33\textwidth]{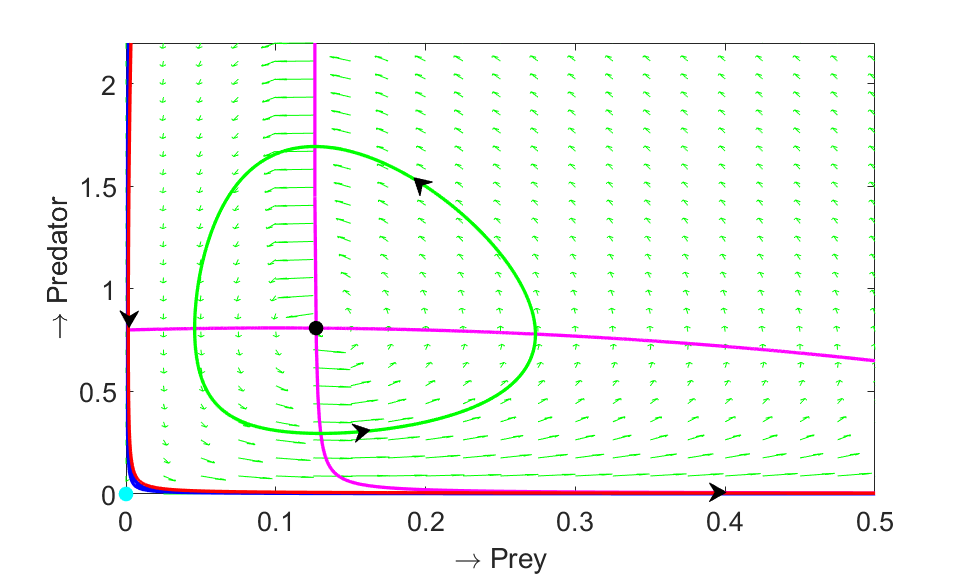}}
\subfigure[]{\includegraphics[width=0.33\textwidth]{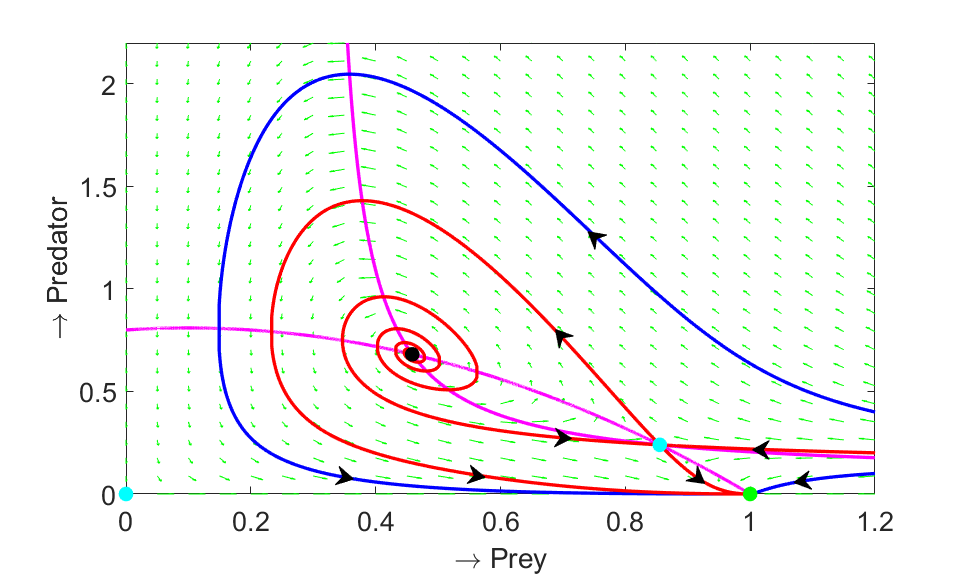}}}
\mbox{\subfigure[]{\includegraphics[width=0.33\textwidth]{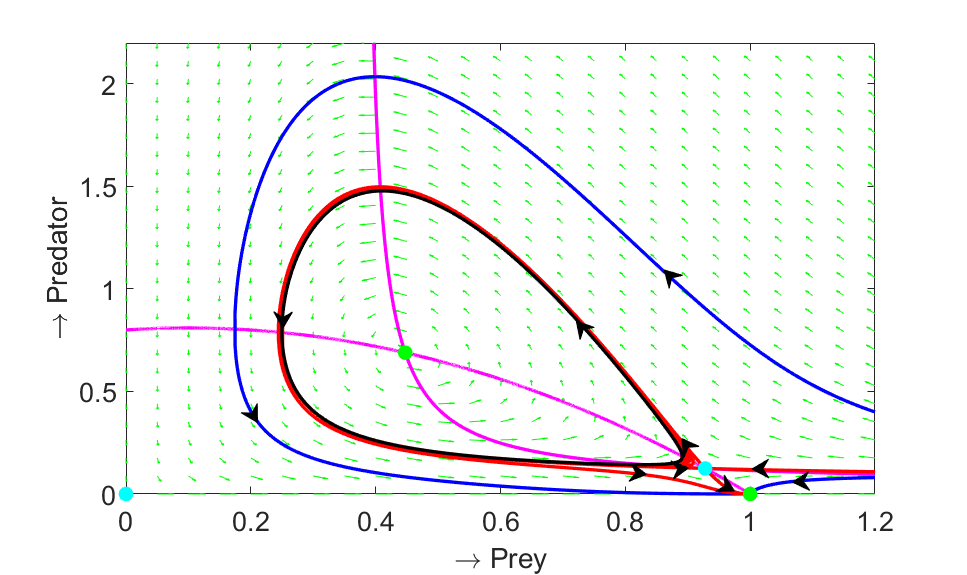}}
\subfigure[]{\includegraphics[width=0.33\textwidth]{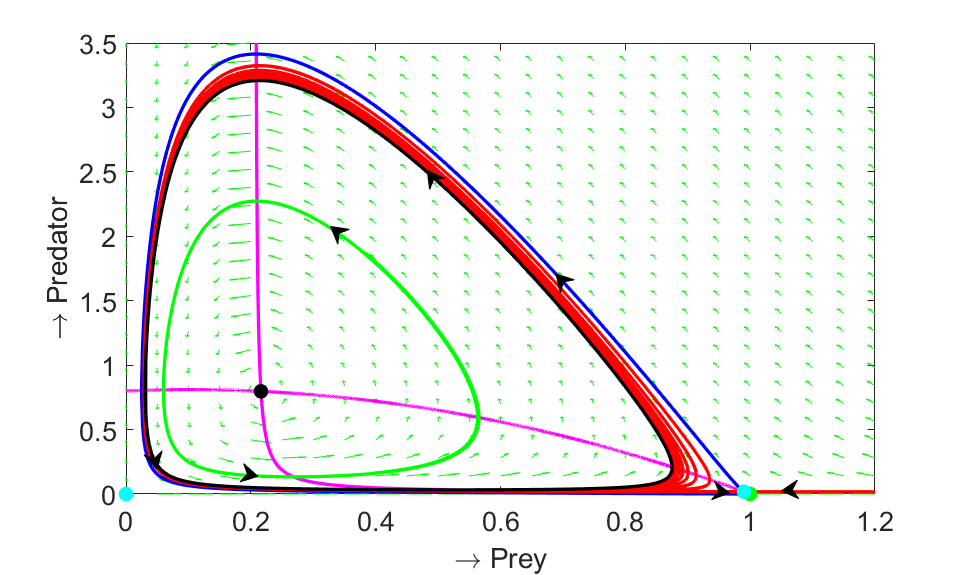}}}
\mbox{\subfigure[]{\includegraphics[width=0.33\textwidth]{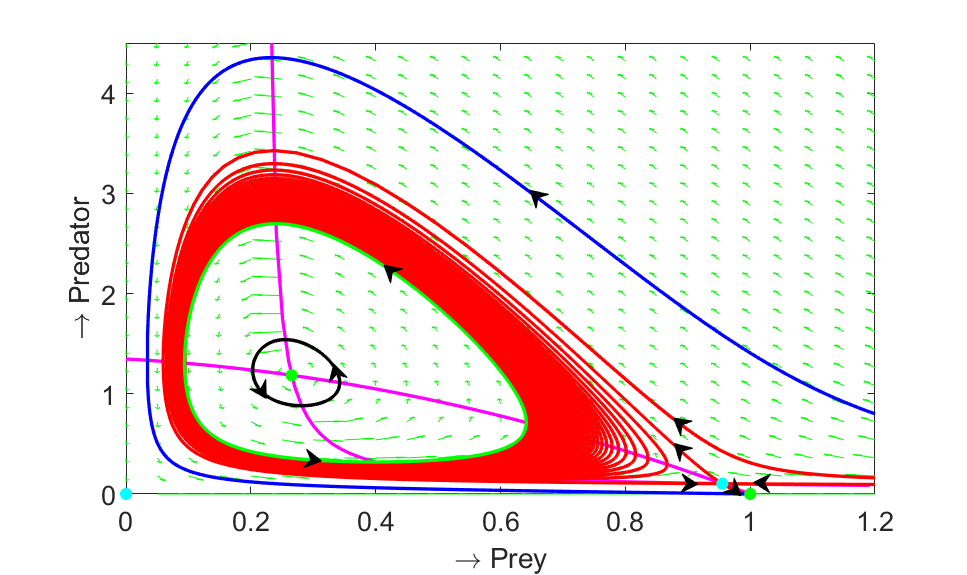}}
\subfigure[]{\includegraphics[width=0.33\textwidth]{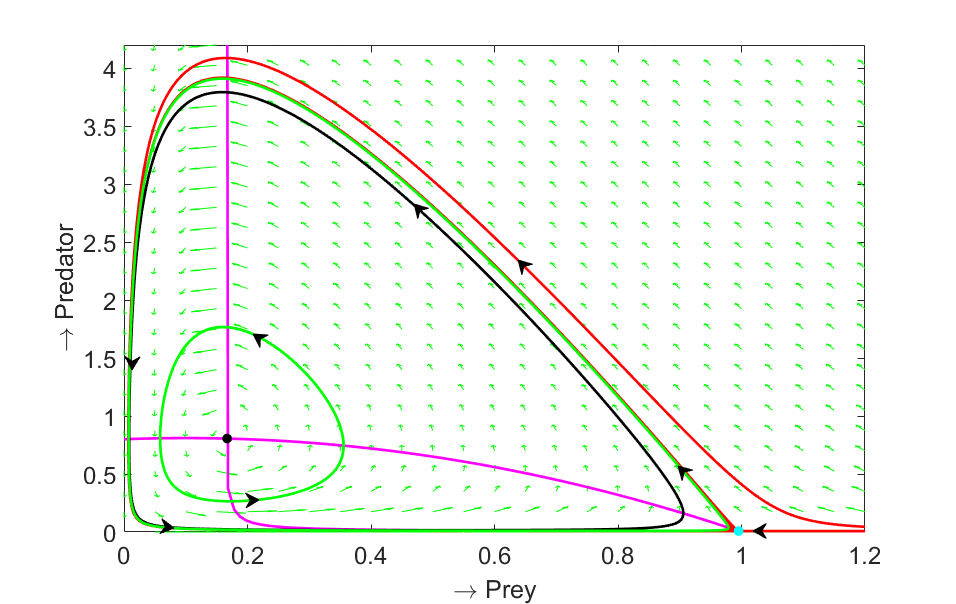}}}
\caption{Phase portraits for the model with the Monod form of
$h(y)$. (a)-(f) are constructed for  $\alpha\,=\,7.1$,
$\beta\,=\,0.8$; the other parameters are (a) $\delta\,=\,0.12$,
$m\,=\,2.66$, region ($R_1$); (b) $\delta\,=\,0.1$, $m\,=\,2.6$,
region ($R_2$); (c) $\delta\,=\,0.01$, $m\,=\,0.96$, region ($R_3$);
(d) $\delta\,=\,0.2$, $m\,=\,2$, region ($R_4$); (e)
$\delta\,=\,0.085$, $m\,=\,2.67$, region ($R_5$); (f)
$\delta\,=\,0.03$, $m\,=\,1.456$, region ($R_6$).
(g) is plotted for $\alpha=7.1$; $\delta=0.2$; $\beta=1.347$; $m=1$,
region ($R_7$); (h) is plotted for $\alpha=6.6$, $\beta=0.8$,
$\delta=0.02$, $m=1.09$,  region ($R_8$).}
\label{phaseportraitsHolling}
\end{figure}


Finally, we consider the case where $\beta$ gradually decreases and the other parameters are kept fixed. This corresponds to the ecologically important scenario in which the environment undergoes gradual eutrophication: an increase in the carrying capacity $K$ in the original model (\ref{eq:1}) corresponds to a proportional decrease in $\beta$ in the dimensionless model. An example of a diagram in the $\beta-m$ plane is constructed for the Monod parametrisation of $h(y)$ (Fig. \ref{Fig_bifdiagSaturating} d). The parameter regions in the diagram have the same meanings as in Fig.\ref{phaseportraitsHolling}. A gradual decrease of $\beta$ results in destabilisation of the coexistence equilibrium and a further collapse of the population of predator (transition from region $R_2$ to region $R_4$). Thus in a eutrophic environment, the only stable equilibrium in the model is $E_{1}$ being a predator free equilibrium. Note that in this diagram we have a new region denoted as $R_7$ in which a locally stable equilibrium $E_{1*}$ is surrounded by two limit cycles: the inner cycle is unstable whereas the outer one is stable. The corresponding phase portrait is shown in Fig.\ref{phaseportraitsHolling}(g). Similar bifurcation behavior is observed for the other two parameterisations of $h(y)$.

\section{Structural sensitivity of the model}

An important part of our investigation is exploring the dependence of model behaviour on the choice of parametrisation of the Allee effect given by $h(y)$. In the previous section, we show that the skeleton of the bifurcation diagram is topologically robust to the mathematical shape of $h(y)$ (Fig.\ref{Fig_3D_diagrams}) when we use three different parameterisations given by the Monod, Ivlev and hyperbolic tangent functions. The relative positions of bifurcation surfaces and possible dynamical regimes remain the same. On the other hand, we also find that for a fixed set of parameters the parametric diagrams can differ considerably, even for close functions $h(y)$. This property is known as structural sensitivity of biological models.

\begin{figure}[ht]
\centerline{
\includegraphics[scale=0.4]{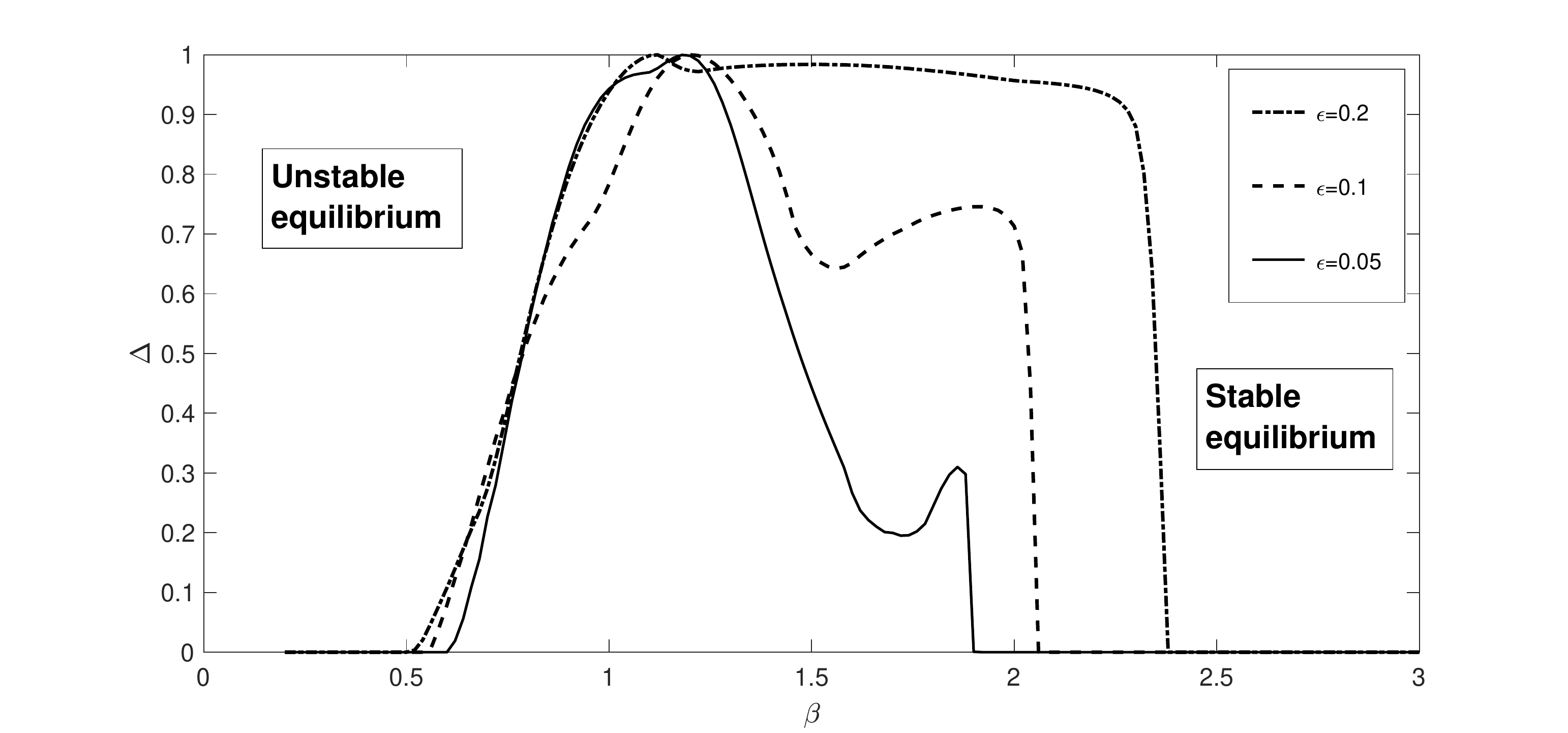}} \caption{The degree of structural sensitivity $\Delta$ in model (\ref{eq:3}) to variation of the parameterisation of the Allee effect in the predator $h(y)$, shown for different $\beta$. The $\epsilon$-neighbourhood of $h(y)$ is defined using the relative difference between the base function and its perturbations. Destabilisation of the equilibrium for the base function (Monod parameterisation) occurs at $\beta=0.975$  For further detail see the text and also \cite{adamson2013can,adamson2014bifurcation}. The other parameters are $\alpha= 7.1$, $\delta=0.12$, $m=1.8$ and $D=10$.} \label{fig:Structural}
\end{figure}

Structural sensitivity of models is an important issue in ecological modelling with a large number of insightful examples provided \cite{adamson2013can,adamson2014bifurcation,adamson2014defining,aldebert2019three,flora2011structural,wood1999super}. The biological rationale behind this idea is that any realistic ecological dependence $h(y)$ can be some combination of the three parameterisations considered above, or, more generally, any other mathematical functions. Here we explore the sensitivity of model predictions to a small variation of $h(y)$ by considering the entire set of functions which satisfy assumptions (A1)-(A5), and we use the methodology from \cite{adamson2013can,adamson2014bifurcation}. Here we investigate the sensitivity of the stability of the interior non-saddle stationary state to small but finite perturbation of $h(y)$ starting from one of the functional forms considered above. Mathematically, we consider the following situation.

Let the base function $h(y)$ be of Monod type. We consider small deviations $h_1(y)$ from the base function $h(y)$ such that $h(y)(1-\epsilon) \le h_1(y)\le h(y)(1+\epsilon)$ and the second derivative of $h_1(y)$ is negative and bounded, $-D< h''_1(y) \le0$, $D>0$ (note that the base function also satisfies this condition). For simplicity we slightly relax our assumption (A5) and allow $h_1(y) \rightarrow 1+\epsilon$, $y \rightarrow \infty$, although this fact is not crucial for the general outcome and conclusions. We conduct the sensitivity analysis in the same way as when investigating the role of the functional response of the predator in stability dynamics in the Rosenzweig-MacArthur model \cite{adamson2013can}. To quantify the structural sensitivity to the choice of $h_1(y)$ we use the degree of structural sensitivity $\Delta$ introduced in \cite{adamson2013can}. $\Delta$ represents the probability that for two randomly chosen functions of $h_1(y)$, the coexistence stationary state will have the same stability properties \cite{adamson2013can}, to be stable or unstable. The maximal degree of sensitivity possible in the system is equal to one, signifying maximal uncertainty in the system


Fig.\ref{fig:Structural} shows the degree of sensitivity $\Delta$  as the parameter $\beta$ is  varied. From the figure one can conclude that a decrease in $\beta$ eventually results in destabilization of the system for any parametrisation of $h_1(y)$: we gradually move from the region of stability to the region of instability shown in the figure. However, depending on the particular choice of $h_1(y)$, this destabilisation may occur within a wide range of $\beta$, thus the system has a plasticity to resist destabilisation caused, for example, by eutrophication (a high value of $K$ in the original system signifies a low $\beta$ in the dimensionless model). Another important observation is that the system exhibits large uncertainty even if the deviation $\epsilon$ from the base function is small ($<5$\%). Considering the Ivlev and the hyperbolic tangent as the base function provides similar results.

\section{Discussion}

The role of the Allee effect in population dynamics has been largely addressed in both empirical and theoretical literature. Surprisingly enough, there has been almost no thorough mathematical investigation into the bifurcation structure of any predator-prey model with an Allee effect in the predator, in particular, this concerns the realistic scenario, where the Allee effect is included in the numerical response of the predator without affecting its functional response \cite{berec2007multiple,courchamp2008allee,dennis1989allee}. This is in a striking contrast to the situation with single species population models or classical predator-prey models with an Allee effect in the prey growth, which have been discussed in all detail and are now included in standard student textbooks in mathematical biology \cite{kot2001elements}. The current study is intended to partially bridge the existing gap. Importantly, we our results are not based on a particular mathematical formulation of the function describing the Allee effect, rather we consider various parameterisations of $h(y)$ which satisfy only few qualitative constraints (A1)-(A5). We have also addressed (for the first time) the issue regarding the sensitivity of the model with respect to parameterisation of the Allee effect (known as the structural sensitivity).

In this section, we will mostly focus on ecological implications of the mathematical results obtained in the previous sections. We should stress, however, that it is crucial to define the way of how we should formally assess the consequences of the Allee effect on the population dynamics. Indeed, this question is far to be a trivial one since distinct paradigms exist in the literature \cite{alves2017hunting,berec2007multiple,courchamp2008allee,dennis1989allee}. In fact, evaluation of the role of the Allee effect in population success should largely depend on the choice of the initial density, and this fact is still somehow disregarded in the literature. Indeed, consider two populations, where one possesses a self-accelerating per capita reproduction rate, and the other one which is characterised by a constant per capita growth rate: for simplicity we assume the growth rates of both populations at the some low density to be the same. Then an increase in the population density would result in an increase in the per capita growth rate of the species with an Allee effect, and this will clearly indicate the benefits of possessing an Allee effect. For example, hunting cooperation is considered to be beneficial for predators up to certain level of population density \cite{alves2017hunting}. On the other hand, considering higher population densities as the starting point for comparison (e.g. population densities where the Allee effect is not pronounced) will show a different outcome. For the same initial per capita growth rates, the species without an Allee effect will be more advantageous since a sudden drop of the population size would not largely affect its per capita growth whereas the population with an Allee effect may exhibit a significant decline in reproduction rate with a threat of extinction. For example, low fertilization efficiency, a lack of mating partners, and sperm limitation are usually considered as negative and undesirable features for population persistence \cite{berec2007multiple,courchamp2008allee,dennis1989allee}. 

Arguably, for many species, their reproduction rate is often empirically estimated at densities which are away from the extinction threshold, where the Allee effect is not well-pronounced. As such, we suggest that in theoretical models the impact of the Allee effect on dynamics should be assessed via comparison with a scenario without an Allee effect, where for both scenarios per capita reproduction rates are assumed to be the same at some `safe' densities. For the current theoretical study, this mathematically signifies that we need to compare the Allee effect in model (\ref{eq:3}) with the same model with $h(y)\equiv 1$ since for large $y$ the value of  $h(y)$ tends to unity, which corresponds to the classical Rosenzweig-MacArthur predator-prey model \cite{kot2001elements}. Note that the model with the Allee effect becomes the Rosenzweig-MacArthur model in the case $\delta \rightarrow 0+$.


Following the above philosophy, our first important conclusion is that introducing the Allee effect in predator's numerical response generally acts as a destabilising factor of a stable coexistence of the prey and the predator in the case where the predator is specialist. Moreover, the Allee effect can result in extinction of predators regardless of initial density. This can be seen from the bifurcation diagrams, when the parameter $\delta$, characterising the strength of the Allee effect, increases from small values $\delta \ll 1$, corresponding to the Rosenzweig-MacArthur predator-prey model, to some large values. Interestingly, destabilising influence of the Allee effect in the predator is observed even for a linear functional response as well (see supplementary material SM4 for detailed illustration). It is well-known that the coexistence equilibrium of the classical Rosenzweig-MacArthur predator-prey model with a linear functional response and a constant $h(y)=const$ is globally stable. Introducing an Allee effect into the predator growth results in stability loss, with either generating sustained oscillations (in this case both prey and predator still persist in the system in an oscillatory mode) or leading to extinction of predator (via different scenarios) for any initial population density. Destabilisation of the system is facilitated with a pronounced saturation in the functional response of the predator (small $\beta$) and with an increase in the strength of the Allee effect - interpreted in the model as a gradual increase in $\delta$. Biologically, this signifies that having  density-dependent $h(y)$ rather than a constant $h(y)=const$ efficiency of predator impedes control over the prey population: the response of the predator to variation of the prey density in the system with the Allee effect becomes delayed and this allows the prey to escape from the control. 

The predicted by the model destabilising role of the Allee effect demonstrates that cyclic population dynamics should occur in predator-prey systems more frequently as it was suggested earlier. For example, in the case of a Holling type I functional response, which is well-known to be stabilising, empirically observed oscillations of population numbers are usually attributed to factors as environmental/demographic noise, seasonal forcing, complex age structure of the population, complexity of the food webs, etc (see \cite{barraquand2017moving} and the references therein). On the other hand, an Allee effect can be an alternative explanation of oscillatory  dynamics of a large number of case studies which occur in non-eutrophic environments with a stabilising functional response of the predator/consumer.

Our second important conclusion is that destabilisation of the coexistence equilibrium in the system with an Allee effect in the predator can occur not only via a supercritical Hopf bifurcation scenario (appearance of a small-amplitude stable limit cycle) but via a subcritical one. In the latter case, destabilisation of the equilibrium will leads to eventual extinction of the predator since the only possible attractor of the model is the state with only prey population being present (regime $R_4$ in the model, Fig.\ref{phaseportraitsHolling}). On the other hand, eutrophication of the environment - which in the dimensionless model corresponds to a decrease in the parameter $\beta$ - would eventually result in extinction of the predator regardless of the scenario of stability loss of the equilibrium (supercritical and subcritical). Indeed, in the case of a supercritical Hopf bifurcation, the resultant predator-prey cycle grows in size and enters the basin of attraction of the prey-only equilibrium: mathematically this occurs via a homoclinical bifurcation.

Our results may have important implications for the biological control of the pests by predators and parasitoids. It has been reported that in a large number of cases,  biological control agents have failed to get established even if under laboratory conditions they could survive by consuming target pest species \cite{bellows2001restoring,roderick2003genes,orr2009biological,bompard2013host}. A possible explanation is the presence of an Allee in the biological control agents which becomes more pronounced in the environment as in a lab. For example, the initial density of the predator can quickly fall below the critical threshold because of dispersal and diffusion. However, there can be a more complicated scenario where the initial density of the predator can be very high, but a pronounced Allee effect will still not allow a long-term persistence of species. Our model mathematically describes this as a globally unstable co-existence state (regime $R_4$ in the model, Fig.\ref{phaseportraitsHolling}). As a conclusion, the choice of the appropriate species for an efficient biological control should be made carefully. For example, for parasitoids it is preferred to use haplodiploid spices (i.e., where the males are haploid and females diploid) to alleviate the negative demographic consequences of mate-finding Allee effects, which is well-pronounced in diploid species \cite{hopper1993mate,bompard2013host}.


Thirdly, we found that inclusion of the Allee effect in the predator growth largely increases the complexity of the system as compared to the scenario without an Allee effect, i.e., the original Rosenzweig-MacArthur predator-prey model. One of the most interesting observation is the possibility of non-trivial multiple attractors, which can be interpreted as alternative ecosystem states. The model predicts two such regimes denoted by $R_7$ and $R_8$, see Fig.\ref{phaseportraitsHolling}. They are (i) a stable equilibrium  coexisting with a stable limit cycle and (ii) two stable limit cycles, respectively. Note that in the literature, there is an ongoing debate on possible origins of alternative states generated through various ecological mechanisms \cite{schroder2005direct}. In the simplest case, alternative stationary states are two contrasting equilibria \cite{scheffer2003floating,scheffer2003catastrophic}. However, more interesting patterns with non-equilibrium coexisting attractors have been reported in the literature as well. In particular, empirical observations demonstrate the possibility of alternative attractors where depending on initial condition the population dynamics can show either cyclic or stable coexistence scenario \cite{zamamiri2001multiple}. Other empirical studies demonstrate the possibility of coexisting cyclic oscillations with contrasting amplitude and periodicity \cite{mccauley1999large,henson2002basins}. Interestingly, in the Daphnia-algae predator-prey system reported in \cite{mccauley1999large}, the observed coexisting cycles were considered to be a consequence of variability of the available food for zooplankton (e.g. a more complicated prey-dependent functional response of the predator or food-dependent efficiency of reproduction), whereas the density dependence of Daphnia's vital rates was somewhat intentionally disregarded. On the other hand, the presence of an Allee effect (e.g. an emergent Allee effect is known to be present in Daphnia \cite{de2003emergent}) can be arguably an alternative explanation of the co-existing cyclic behaviour. 

Finally, we should stress that unlike the scenario with the Allee effect in prey - which is currently considered to be straightforward in the literature - including the Allee effect in predators can be somewhat tricky since this can be done either by modifying the functional response or the numeric response of the predator. Biologically, this signifies that distinct mechanisms of emerging  the resultant  demographic Allee effect - a decrease in reproduction at low population numbers -  should be modelled in a different way. In particular, the mechanisms such as low fertilization efficiency, a lack of mating partners, sperm limitation, cooperative breeding and similar mechanisms (see \cite{berec2007multiple,courchamp2008allee,dennis1989allee}) should be included in the numerical response of predator only, whereas collective exploitation of a resource such as cooperating hunting should be included in both functional and the numerical responses. It is rather surprising that this fact has not been largely explored in the literature yet.

As a first step to fill the existing gap, we compared the stability of two similar predator-prey systems: in one model the Allee effect was due to collective exploitation of resources and in the other one the Allee effect was due to the lack of mating partners (for details see SM4). Note that unlike the influential study \cite{alves2017hunting} we included saturation in the Allee effect in the functional response of predator (see the model equations in SM4). We found that a pronounced Allee effect (large saturation in $h(y)$) has different consequences for the two systems. In particular, an Allee effect due to collective exploitation of resources seems to facilitate  persistence of the predator as compared to the scenario where the Allee effect is included only in numerical response (e.g. due to mate-finding). However, a more detailed comparison of the two mentioned approaches to modelling the Allee effect in predator should be an interesting separate study.

\bibliographystyle{abbrv}
\bibliography{ref}  

\section*{Supplementary Material 1 (SM1)}
\noindent \textbf{Proof of Proposition 2}: Suppose that the interior equilibrium undergoes a generalised Hopf (GH) bifurcation at $\hat{E}\,=\,(\hat{x},\hat{y})$ for the parameter threshold $(\beta_{GH},m_{GH})$. Since $\hat{E}$ is a non-trivial equilibrium point then it must satisfies (3). The Jacobian matrix  at $\hat{E}$ is given by

\begin{eqnarray*}
\label{Gh_Jacobian}
J(\hat{E}) &=& \begin{bmatrix}
\frac{\hat{x}}{\beta + \hat{x}}\left(1-\beta-2\hat{x}\right) & -\frac{\hat{x}}{\beta + \hat{x}}\\[0.5em]
\frac{\alpha \beta \hat{y}h(\hat{y})}{(\beta+\hat{x})^2} & \frac{\alpha \hat{x}\hat{y}}{\beta+\hat{x}}h^{'}(\hat{y})
\end{bmatrix}\quad.
\end{eqnarray*}

\noindent  Now using the condition (GH2), we find from the above Jacobian

\begin{eqnarray*}
\label{eq:Gh2}
\beta_{GH} &=& 1-2\hat{x}+\alpha \hat{y}h^{'}(\hat{y}).
\end{eqnarray*}

\noindent Now to find the first Liapunov number at $\hat{E}$, we translate $\hat{E}$ to origin by using the transformation $u\,=\,x-\hat{x}$, $v\,=\,y-\hat{y}$ and we get

\begin{subequations}
\label{eq:expand}
\begin{align*}
\begin{split}
\dot{u} &= au + bv + P(u,v),
\end{split}\\[0.5em]
\begin{split}
\dot{v} &= cu + dv + Q(u,v),
\end{split}
\end{align*}
\end{subequations}

\noindent where $a\,=\,\left. \frac{\partial F_{1}}{\partial x} \right|_{\hat{E}}$, $b\,=\,\left. \frac{\partial F_{1}}{\partial y} \right|_{\hat{E}}$, $c\,=\,\left. \frac{\partial F_{2}}{\partial x} \right|_{\hat{E}}$, $d\,=\,\left. \frac{\partial F_{2}}{\partial y} \right|_{\hat{E}}$, $P(u,v)$ and $Q(u,v)$ are analytic function is given by

\begin{eqnarray*}
P(u,v) &=& \sum_{i+j\geq 2} a_{ij}u^{i}v^{j},\\[0.5em]
Q(u,v) &=& \sum_{i+j\geq 2} b_{ij}u^{i}v^{j},
\end{eqnarray*}

\noindent where $a_{ij}$ and $b_{ij}$ are obtained from the following relations

\begin{subequations}
\label{eq:Fcoeffi}

\begin{align*}
a_{ij} \,=\, \frac{1}{i!j!} \left. \frac{\partial^{i+j} F_{1}}{\partial x^{i}\partial y^{j}}\right|_{\hat{E}},\,\,\,
b_{ij} \,=\, \frac{1}{i!j!} \left. \frac{\partial^{i+j} F_{2}}{\partial x^{i}\partial y^{j}}\right|_{\hat{E}}.
\end{align*}

\end{subequations}

\noindent Now first Liapunov number is defined by

\begin{align*}
\label{eq:LiapunovNo}
\begin{split}
l &= -\frac{3\pi}{2b\Delta^{\frac{3}{2}}}\left[\{ac(a_{11}^{2}+a_{11}b_{02}+a_{02}b_{11})+ab(b_{11}^2+a_{20}b_{11}+a_{11}b_{02}-2acb_{02}^2
\right.\\[0.5em]
 & \left. -2ab(a_{20}^2-b_{20}b_{02})-b^2(2a_{20}b_{20}+b_{11}b_{20})+(bc-2a^2)(b_{11}b_{02}-a_{11}a_{20})\}\right.\\
 & \left. -(a^2+bc)\{3(cb_{03}-a_{30})+2a(a_{21}+b_{12})+(ca_{12}-bb_{21})\}\right].
\end{split}
\end{align*}

\noindent We calculate the coefficients $a_{ij}$, $b_{ij}$ and $a,b,c,d$ of the first Lyapunov number at $\hat{E}$:

\begin{eqnarray*}
a &=& \frac{\hat{x}}{\beta + \hat{x}}(1-\beta -2\hat{x}), \quad b\,=\,-\frac{\hat{x}}{\beta + \hat{x}}, \quad c\,=\,\frac{\alpha \beta \hat{y}h(\hat{y})}{(\beta + \hat{x})^2}, \quad d\,=\,\frac{\alpha \hat{x}\hat{y}h^{'}(\hat{y})}{\beta + \hat{x}},
\end{eqnarray*}

\begin{eqnarray*}
a_{20} &=& \left(\frac{\beta \hat{y}}{(\beta + \hat{x})^3}-1\right), \quad a_{11}\,=\,-\frac{\beta}{(\beta+\hat{x})^2}, \quad a_{02}\,=\,0,
\end{eqnarray*}

\begin{eqnarray*}
b_{20} &=& -\frac{\alpha \beta \hat{y}h(\hat{y})}{(\beta + \hat{x})^3}, \quad b_{11}\,=\,\frac{\alpha \beta}{(\beta+\hat{x})^2}\left[h(\hat{y})+\hat{y}h^{'}(\hat{y})\right], \quad b_{02}\,=\,\frac{\alpha \hat{x}}{2(\beta+\hat{x})}\left[2h^{'}(\hat{y})+\hat{y}h^{''}(\hat{y})\right],
\end{eqnarray*}

\begin{eqnarray*}
a_{30} &=& -\frac{\beta \hat{y}}{2(\beta+\hat{x})^4}, \quad a_{21}\,=\,\frac{\beta}{(\beta + \hat{x})^3},\quad a_{12}\,=\,0, \quad a_{03}\,=\,0,
\end{eqnarray*}

\begin{eqnarray*}
b_{30} &=& \frac{\alpha \beta \hat{y}h(\hat{y})}{(\beta + \hat{x})^4}, \quad b_{21}\,=\,-\frac{\alpha \beta}{(\beta + \hat{x})^3}\left[h(\hat{y})+\hat{y}h^{'}(\hat{y})\right],
\end{eqnarray*}

\begin{eqnarray*}
b_{12} &=& \frac{\alpha \beta}{2(\beta + \hat{x})^2}\left[2h^{'}(\hat{y})+\hat{y}h^{''}(\hat{y})\right], \quad b_{03}\,=\,\frac{\alpha \hat{x}}{6(\beta + \hat{x})}\left[3h^{''}(\hat{y})+\hat{y}h^{'''}(\hat{y})\right].
\end{eqnarray*}

\noindent After some algebraic calculation we get (assuming that $h(y)$ is as smooth as that $h^{'''}(y)$ exist)

$$l(\hat{E};\beta\,=\,\beta_{GH},\,m\,=\,m_{GH}))\,=\, -\frac{3\Pi}{2b\Delta_{GH}^{\frac{3}{2}}}\left[Am_{GH}^2+Bm_{GH}+C\right],$$

\noindent where

\begin{eqnarray*}
A &=& \left.-\frac{\alpha \beta^2 m^2 y}{(\beta + x)^4}\left[2h^{'}(y)+\frac{y}{2}h^{''}(y)\right]\right|_{(\hat{E};\beta\,=\,\beta_{GH},m\,=\,m_{GH})},
\end{eqnarray*}

\begin{eqnarray*}
B &=& \left[\frac{\beta(1-\beta-2x)}{(\beta+x)^3}\{\frac{\beta^2 y}{(\beta+x)^3}-\frac{2\alpha \beta xyh^{'}(y)}{(\beta+x)^2}-\frac{x}{\beta+x}(\frac{\beta y}{(\beta+x)^3}-1)\}\right.\\
 & & \left.-\frac{\alpha \beta x}{(\beta+x)^3}(2h^{'}(y)+yh^{''}(y))\{\frac{\beta y(1-\beta - 2x)}{2(\beta + x)^2}+\frac{\alpha xy}{2}(1-\beta-2x)\right.\\
 & & \left.(2h^{'}(y)+h^{''}(y))+\frac{xy}{(\beta + x)^2}(1-\beta-2x)-\frac{\alpha \beta y^2 h^{''}(y)}{2(\beta+x)^3}-\frac{2x(1-\beta-2x)}{\beta + x}\}\right.\\
 & & \left.\left.-\frac{\beta x}{(\beta + x)^4}\{x(1-\beta-2x)^2-\frac{\alpha \beta yh(y)}{\beta + x}\}\{\frac{\alpha y}{2}(3h^{''}(y)+yh^{'''}(y))-\frac{1}{\beta+x}\}\right]\right|_{(\hat{E};\beta\,=\,\beta_{GH},m\,=\,m_{GH})},
\end{eqnarray*}

\begin{eqnarray*}
C &=& -\frac{x^2}{(\beta + x)^2}(1-\beta-2x)\left[\frac{\alpha^2\beta^2y^2(h^(y))^2}{(\beta+x)^4}+\frac{\alpha \beta yh^{'}(y)}{(\beta+x)^2}\left(\frac{\beta y}{(\beta+x)^3}-1\right)-\frac{\alpha \beta x}{2(\beta+x)^3}\right.\\
& & \left.(2h^{'}(y)+yh^{''}(y))+2\left(\frac{\beta y}{(\beta+x)^3}-1\right)^2+\frac{2\alpha^2\beta xy(1-\beta-2x)h^{'}(y)}{(\beta+x)^3}-\frac{x}{(\beta+x)^5}\right.\\
& & \left. \{x(1-\beta-2x)^2-\frac{\alpha \beta yh(y)}{\beta+x}\}\{-\frac{3\beta xy}{(\beta+x)^2}\right.\\
& & \left.\left.+2\beta x(1-\beta -2x)(\frac{1}{\beta+x}+\frac{\alpha}{2}(2h^{'}(y)+yh^{''}(y)))-\frac{\alpha \beta xyh^{'}(y)}{\beta+x}\}\right]\right|_{(\hat{E};\beta\,=\,\beta_{GH},m\,=\,m_{GH})}.
\end{eqnarray*}

\noindent  Due to algebraic complexity we are not able to show that $l((\hat{E};\beta\,=\,\beta_{GH},m\,=\,m_{GH}))\,=\,0$,  however it can be verified numerically. Thus we can conclude that the system (2) admits GH bifurcation for the parameter threshold $(\beta\,=\,\beta_{GH},m\,=\,m_{GH})$ at $\hat{E}$.\\[1em]

\section*{Supplementary Material 2 (SM2)}

\noindent \textbf{Proof of Proposition 3:} The conditions (BT1) and (BT2) are equivalent to the following ones

\begin{eqnarray*}\label{eq:Bt_det}
\left. \bar{x}h^{'}(\bar{y})(1-\beta-2\bar{x})+\frac{\beta
h(\bar{y})}{\beta+\bar{x}}\right|_{(\beta_{BT},m_{BT})}\,=\,0,
\end{eqnarray*}

\noindent and

\begin{eqnarray*}\label{eq:Bt_trace}
\left.1-\beta-2\bar{x}+\alpha \bar{y}
h^{'}(\bar{y})\right|_{(\beta_{BT},m_{BT})} &=& 0.
\end{eqnarray*}

\noindent It is difficult to find the explicit expressions of the bifurcation parameters for the BT bifurcation thresholds since in the above equations contain unknown function both $h(y)$ and the coordinates of the equilibrium points contain the parameters implicitly. However, we can define the bifurcation parameters as follows:

\begin{eqnarray*}
\beta_{BT} &=& 1-2\bar{x}+\alpha \bar{y}h^{'}(\bar{y}),\\[0.5em]
m_{BT} &=&
\alpha\left[\bar{x}h^{'}(\bar{y})(1-\beta_{BT}-2\bar{x})+h(\bar{y})\right].
\end{eqnarray*}

\noindent Let us consider the small perturbation around the bifurcation threshold $(\beta_{BT},m_{BT})$ by $(\beta_{BT}+\lambda_{1},m_{BT}+\lambda_{2})$, where $\lambda_{i},\,i\,=\,1,2$ are sufficiently small, and substituting it into the system we get

\begin{subequations} \label{eq:BT1}
\begin{align*}
\begin{split}
\frac{dx}{dt}&=x\left(1-x\right)-\frac{xy}{\beta_{BT} + \lambda_{1} + x}\,\equiv\,F_{1}(x,y,\lambda_{1}),
\end{split}\\
\begin{split}
\frac{dy}{dt}&=\frac{\alpha xy}{\beta_{BT} + \lambda_{1} + x}h(y) - \left(m_{BT}+\lambda_{2}\right)y\,\equiv\,F_{2}(x,y,\lambda_{1},\lambda_{2}),
\end{split}
\end{align*}
\end{subequations}

\noindent Now we shift the equilibrium $\bar{E}\,=\,(\bar{x},\bar{y})$ to the origin by the coordinate transform $X\,=\,x-\bar{x}$, $Y\,=\,y-\bar{y}$ and substitute it in model (2) to obtain

\begin{subequations}
\label{eq:BT2}
\begin{align*}
\begin{split}
\dot{X} &= F_{10} + \tilde{a}X + \tilde{b}Y + p_{20}X^2 + p_{11}XY + H_{1}(X,Y),
\end{split}\\[0.5em]
\begin{split}
\dot{Y} &= F_{20} + \tilde{c}X + \tilde{d}Y + q_{20}X^2 + q_{11}XY + q_{02}Y^2 + H_{2}(X,Y),
\end{split}
\end{align*}
\end{subequations}

\noindent where

\begin{eqnarray*}
F_{10} &=& \bar{x}\left(1-\bar{x}\right)-\frac{\bar{x}\bar{y}}{\beta_{BT} + \lambda_{1} + \bar{x}},\quad \tilde{a} \,=\, 1-2\bar{x}-\frac{\bar{y}}{\beta_{BT} + \lambda_{1} + x} + \frac{\bar{x}\bar{y}}{(\beta_{BT} + \lambda_{1} + x)^2},\\[0.2em]
\tilde{b} &=& -\frac{\bar{x}}{\beta_{BT} + \lambda_{1} + \bar{x}},\quad p_{20}\,=\, -\left(1-\frac{\bar{y}}{\beta_{BT} + \lambda_{1} + \bar{x}}+\frac{\bar{x}\bar{y}}{(\beta_{BT} + \lambda_{1} + \bar{x})^2}\right),\\[0.3em]
p_{11} &=& -\frac{\beta_{BT}+\lambda_{1}}{(\beta_{BT}+\lambda_{1}+\bar{x})^2},\quad F_{20}\,=\,\frac{\alpha \bar{x}\bar{y}}{\beta_{BT} + \lambda_{1} + \bar{x}}h(\bar{y}) - \left(m+\lambda_{2}\bar{y}\right),\\[0.3em]
\tilde{c} &=& \frac{\alpha(\beta_{BT}+\lambda_{1})\bar{y}h(\bar{y})}{(\beta_{BT}+\lambda_{1}+\bar{x})^2},\quad \tilde{d}\,=\, \frac{\alpha\bar{x}}{\beta_{BT}+\lambda_{1}+\bar{x}}\left(h(\bar{y})+\bar{y}h^{'}(\bar{y})\right)-(m_{BT}+\lambda_{2}),\\[0.3em]
q_{20} &=& -\frac{\alpha(\beta_{BT}+\lambda_{1})\bar{y}h(\bar{y})}{(\beta_{BT}+\lambda_{1}+\bar{x})^3},\quad q_{11}\,=\,\frac{\alpha(\beta_{BT}+\lambda_{1})}{(\beta_{BT}+\lambda_{1}+\bar{x})^2}\left[h(\bar{y})+\bar{y}h^{'}(\bar{y})\right],\\[0.3em]
q_{02} &=& \frac{\alpha \bar{x}}{\beta_{BT}+\lambda_{1}+\bar{x}}\left[2h^{'}(\bar{y})+\bar{y}h^{''}(\bar{y})\right],
\end{eqnarray*}

\noindent and $H_{1}$, $H_{2}$ are power series of $X$ and $Y$ with terms $X^{i}Y^{j},\,i+j\geq 3$. Now by using affine transformation $X_{1}\,=\,X$, $X_{2}\,=\,\tilde{a}X+\tilde{b}Y$, the system transferred to

\begin{subequations}
\label{eq:BT3}
\begin{align*}
\begin{split}
\dot{X}_{1} &= F_{10} + X_{2} + l_{20}X_{1}^2+l_{11}X_{1}X_{2} + \tilde{H}_{1}(X_{1},X_{2}),
\end{split}\\[0.5em]
\begin{split}
\dot{X}_{2} &= k_{0} + k_{10}X_{1} + k_{01}X_{2} + k_{20}X_{1}^{2} + k_{11}X_{1}X_{2} + k_{02}X_{2}^{2} + \tilde{H}_{2}(X_{1},X_{2}),
\end{split}
\end{align*}
\end{subequations}

where,

\begin{eqnarray*}
l_{20} &=& p_{20}-\frac{\tilde{a}}{\tilde{b}}p_{11},\quad l_{11}\,=\, \frac{p_{11}}{\tilde{b}},\\[0.3em]
k_{0} &=& \tilde{a}F_{10}+\tilde{b}F_{20},\quad k_{10}\,=\,\tilde{b}\tilde{c}-\tilde{a}\tilde{d},\quad k_{01}\,=\,\tilde{a}+\tilde{d},\\[0.2em]
k_{20} &=& \tilde{a}p_{20}-\frac{\tilde{a}}{\tilde{b}}p_{11} + \tilde{b}q_{20} - \tilde{a}q_{11} + \frac{\tilde{a}^2}{\tilde{b}}q_{02},\\[0.3em]
k_{11} &=& \frac{\tilde{a}}{\tilde{b}}p_{11} + q_{11} - 2\frac{\tilde{a}}{\tilde{b}}q_{02},\quad k_{02}\,=\,\frac{q_{02}}{\tilde{b}},
\end{eqnarray*}

\noindent and $\tilde{H}_{1}$, $\tilde{H}_{2}$ are power series of $X_{1}$ and $X_{2}$ with terms $X_{1}^{i}X_{2}^{j},\,i+j\geq 3$. Now using the $C^{\infty}$ change of coordinates in the neighbourhood of origin by the following transformation

$$v_{1}\,=\,X_{1}-\frac{l_{11}+k_{02}}{2}X_{1}^2, \quad v_{2}\,=\,X_{2}+l_{20}X_{1}^2-k_{02}X_{1}X_{2},$$

\noindent the system reduces to

\begin{subequations}
\label{eq:BT4}
\begin{align*}
\begin{split}
\dot{v}_{1} &= F_{10} + s_{10}v_{1} + v_{2} + s_{20}v_{1}^2+ R_{1}(v_{1},v_{2}),
\end{split}\\[0.5em]
\begin{split}
\dot{X}_{2} &= k_{0} + t_{10}v_{1} + t_{01}v_{2} + t_{20}v_{1}^{2} + t_{11}v_{1}v_{2} + R_{2}(v_{1},v_{2}),
\end{split}
\end{align*}
\end{subequations}

\noindent where

\begin{eqnarray*}
s_{10} &=& -\frac{F_{10}}{\tilde{b}}(p_{11}+q_{02}),\quad s_{20}\,=\,-\frac{F_{10}}{2\tilde{b}^2}(p_{11}+q_{02})^2,\\[0.2em]
t_{10} &=& k_{10}+2F_{10}l_{20}-k_{0}k_{02},\quad t_{01}\,=\,k_{01}-\frac{F_{10}}{\tilde{b}}q_{02},\\[0.2em]
t_{20} &=& \frac{k_{10}}{2}(l_{11}+k_{02})-k_{01}l_{20}+k_{20}+F_{10}l_{20}(l_{11}+k_{02})-k_{01}k_{02}-\frac{k_{0}k_{02}}{2}(l_{11}+k_{02}) + F_{10}l_{20}k_{02},\\[0.2em]
t_{11} &=& k_{01}k_{02}+k_{11}+2l_{20}-k_{01}k_{02}-k_{02}^2F_{10},
\end{eqnarray*}

\noindent and $R_{1}$, $R_{2}$ are power series of $v_{1}$ and $v_{2}$ with terms $v_{1}^{i}v_{2}^{j},\,i+j\geq 3$.

Now again using the $C^{\infty}$ change of coordinates in the neighbourhood of origin by the following transformation

$$z_{1}\,=\,v_{1}, \quad z_{2}\,=\,F_{10}+s_{10}v_{1}+v_{2}+s_{20}v_{1}^2,$$

\noindent the system becomes reduced to

\begin{subequations}
\label{eq:BT5}
\begin{align*}
\begin{split}
\dot{z}_{1} &= z_{2} + \tilde{R}_{1}(z_{1},z_{2}),
\end{split}\\[0.5em]
\begin{split}
\dot{z}_{2} &= w_{0} + w_{10}z_{1} + w_{01}z_{2} + w_{20}z_{1}^{2} + w_{11}z_{1}z_{2} + \tilde{R}_{2}(z_{1},z_{2}),
\end{split}
\end{align*}
\end{subequations}

\noindent where

\begin{eqnarray*}
w_{0} &=& t_{0}-t_{01}F_{10},\quad w_{10}\,=\,t_{10}-t_{01}s_{10}-t_{11}F_{0},\quad w_{01}\,=\,s_{10}+t_{01},\\[0.2em]
w_{20} &=& t_{20}-t_{01}s_{20}-t_{11}s_{10},\quad w_{11}\,=\,t_{11}+2s_{20},
\end{eqnarray*}

\noindent and $\tilde{R}_{1}$, $\tilde{R}_{2}$ are power series of $z_{1}$ and $z_{2}$ with terms $z_{1}^{i}z_{2}^{j},\,i+j\geq 3$. Finally, using the $C^{\infty}$ change of coordinates in the neighbourhood of origin by the following transformation

$$n_{1}\,=\,z_{1},\quad n_{2}\,=\,z_{2}+\tilde{R}_{1}(z_{1},z_{2}),$$

\noindent the system changes to

\begin{subequations}
\label{eq:BT6}
\begin{align*}
\begin{split}
\dot{n}_{1} &= n_{2},
\end{split}\\[0.5em]
\begin{split}
\dot{n}_{2} &= P(n_{1},\lambda_{1},\lambda_{2}) + n_{2}\phi(n_{1},\lambda_{1},\lambda_{2}) + n_{2}^{2}\xi(n_{1},n_{2},\lambda_{1},\lambda_{2}),
\end{split}
\end{align*}
\end{subequations}

\noindent where $P,\,\phi,\,\xi\,\in\,C^{\infty}$ and

\begin{eqnarray*}
P(n_{1},\lambda_{1},\lambda_{2}) &=& w_{0} + w_{10}n_{1} + w_{20}n_{1}^{2} + P_{1}(n_{1}),\\[0.2em]
\phi(n_{1},\lambda_{1},\lambda_{2}) &=& w_{01} + w_{11}n_{1} + \phi_{1}(n_{1}),\\[0.2em]
\xi(n_{1},n_{2},\lambda_{1},\lambda_{2}) &=& \xi_{1}(n_{1},n_{2}).
\end{eqnarray*}

\noindent Here $P_{1}$, $\phi_{1}$, $\xi_{1}$ are the power series in the terms $n_{1}^{i},\,i\,\geq\,3$, $n_{1}^{i},\,i\,\geq\,2$ and $n_{1}^{i}n_{2}^{j},\,i+j\,\geq\,1$ respectively. Now applying Malgrange preparation theorem to the function $P$, we get

$$P(n_{1},\lambda_{1},\lambda_{2})\,=\,\left(n_{1}^2+\frac{w_{10}}{w_{20}}n_{1}+\frac{w_{0}}{w_{20}}\right)B(n_{1},\lambda_{1},\lambda_{2}),$$

\noindent where $B$ is a power series of $n_{1}$ and $B(0,\lambda_{1},\lambda_{2})\,=\,w_{20}$. Now using the transformation $u_{1}\,=\,n_{1}$, $u_{2}\,=\,\frac{n_{2}}{\sqrt{B(n_{1},\lambda_{1},\lambda_{2})}}$ and $T\,=\,\int_{0}^{t} \sqrt{B(n_{1}(s),\lambda_{1},\lambda_{2})} ds$, the system reduces to

\begin{subequations}
\label{eq:BT7}
\begin{align*}
\begin{split}
\dot{u}_{1} &= u_{2},
\end{split}\\[0.5em]
\begin{split}
\dot{u}_{2} &= \frac{w_{0}}{w_{20}} + \frac{w_{10}}{w_{20}}u_{1} + \frac{w_{01}}{\sqrt{w_{20}}}u_{2} + u_{1}^{2} + \frac{w_{11}}{\sqrt{w_{20}}}u_{1}u_{2} + Q_{1}(u{1},u_{2},\lambda_{1},\lambda_{2}),
\end{split}
\end{align*}

\noindent where $Q_{1}$ is a power series in $(u_{1},u_{2})$ with power $u_{1}^{i}u_{2}^{j},\,i+j\,\geq\,3$ and $j\,\geq\,2$. Now applying the final transformation
\end{subequations}

$$y_{1}\,=\,u_{1} + \frac{w_{10}}{2w_{20}},\quad y_{2}\,=\,u_{2},$$

\noindent the system transforms to

\begin{subequations}
\label{eq:BT8}
\begin{align*}
\begin{split}
\dot{y}_{1} &= y_{2},
\end{split}\\[0.5em]
\begin{split}
\dot{y}_{2} &= \mu_{1}(\lambda_{1},\lambda_{2}) + \mu_{2}(\lambda_{1},\lambda_{2})y_{2} + y_{1}^{2} + \frac{w_{11}}{\sqrt{w_{20}}} + Q_{2}(u{1},u_{2},\lambda_{1},\lambda_{2}),
\end{split}
\end{align*}
\end{subequations}

\noindent where

\begin{eqnarray*}
\mu_{1}(\lambda_{1},\lambda_{2}) &=& \frac{w_{0}}{w_{20}}-\left(\frac{w_{10}}{2w_{20}}\right)^2,\\[2em]
\mu_{2}(\lambda_{1},\lambda_{2}) &=& \frac{w_{01}}{\sqrt{w_{20}}}-\frac{w_{10}w_{11}}{2w_{20}\sqrt{w_{20}}},
\end{eqnarray*}

\noindent and $Q_{2}$ is a power series in $y_{1}$ and $y_{2}$ of the form $y_{1}^{i}y_{2}^{j},\,i+j\,\geq\,3$.

To check the non-degeneracy conditions for BT bifurcation at $\bar{E}$ we need to check the signs of $k_{20}$ and $l_{20}+k_{11}$ for $\lambda_{1}\,=\,0$ and $\lambda_{2}\,=\,0$ i.e. at $\beta\,=\,\beta_{BT}$ and $m\,=\,m_{BT}$. Now

\begin{eqnarray*}
k_{20} &=& aa_{20}-\frac{a^2}{b}a_{11}+bb_{20}-ab_{11}+\frac{a^2}{b}\\[0.5em]
 &=& \left[-\frac{\bar{x}^2}{(\beta_{BT} + \bar{x})^2}\left(1-\beta_{BT}-2\bar{x}\right) + \frac{\alpha \beta_{BT} \bar{x} h(\bar{y})}{(\beta_{BT}+\bar{x})^2} - \frac{\alpha \bar{x}}{(\beta_{BT}+\bar{x})^2}(1-\beta_{BT}-2\bar{x})\right.\\
 & & \left.h^{'}(\bar{y})\left\lbrace \beta_{BT}(1-\bar{x})-2\bar{x}(1-\beta_{BT}-2\bar{x})+\bar{x}\bar{y}(1-\beta_{BT}-2\bar{x})h^{''}(\bar{y})\right\rbrace\right],
\end{eqnarray*}

\noindent where $a$, $b$, $c$, $d$, $a_{ij}$, $b_{ij}$, $i+j\,\geq\,1$ are same as in the Supplementary Material A except one needs to calculate the coefficients at $\bar{E}$. We have $1-\beta_{BT}-2\bar{x} < 0$, as $\alpha$ and $h(y)$ is positive. Using this with the condition (A4) we can find $k_{20} > 0$

\begin{eqnarray*}
l_{20} &=& a_{20}-\frac{a}{b}a_{11}\,=\,-\frac{\bar{x}}{\beta_{BT}+\bar{x}} < 0
\end{eqnarray*}
\noindent and

\begin{eqnarray*}
k_{11} &=& \frac{a}{b}a_{11}+b_{11}-\frac{2a}{b}b_{02}\\[0.5em]
 &=& (1-\beta_{BT}-2\bar{x})\left[\frac{\beta_{BT}}{(\beta_{BT}+\bar{x})^2}+\alpha \bar{x}(1-\bar{x})h''(\bar{y})\right]+\alpha (1-2\bar{x})h'(\bar{y}).
\end{eqnarray*}

\noindent Now $\textup{sign}(k_{11})$ may changes from negative to positive or vice versa. Therefore $l_{20} + k_{11}$ may or may not be zero. If $l_{20} + k_{11}\,\neq\,0$ then the system near the BT bifurcation point is topologically equivalent to the following one
\begin{subequations}
\begin{align*}
\begin{split}
\dot{x}_{1} &= x_{2},
\end{split}\\[0.2em]
\begin{split}
\dot{x}_{2} &= \mu_{1} + \mu_{2}x_{2} +  x_{1}^2 \pm x_{1}x_{2},
\end{split}
\end{align*}
\end{subequations}
\noindent which is the normal form of the BT bifurcation of co-dimension two. The coefficient in front of $x_{1}x_{2}$ can be either $+1$ or $-1$ according as $\textup{sign}(k_{20}(l_{20}+k_{11}))$ . If the sign is positive then the predator-prey model undergoes a subcritical BT bifurcation at $\bar{E}$ and if the sign is negative then the model undergoes supercritical BT bifurcation at $\bar{E}$.

 In the case where $l_{20} + k_{11}\,= 0$, we have a co-dimension 3 Bogdanov-Takens bifurcation and the system is topologically equivalent to the following one
  \begin{subequations}
\begin{align*}
\begin{split}
\dot{x}_{1} &= x_{2},
\end{split}\\[0.2em]
\begin{split}
\dot{x}_{2} &= \mu_{1} + \mu_{2}x_{2} + \mu_{3} x_{1}x_{2}+ x_{1}^2 \pm x_{1}^3 x_{2}.
\end{split}
\end{align*}
\end{subequations}

Note that this type of co-dimension 3 Bogdanov-Takens bifurcation includes a double-equilibrium  point.

\section*{Supplementary Material 3 (SM3)}

Here we resent several $\alpha$ cross sections of the $(\alpha,\delta,m)$ bifurcation diagrams shown in Figs. 6 am 7 of the main text. We consider the case where parametrization of $h(y)$ is given either the Ivlev or the trigonometric response.

\begin{figure}[H]
\centerline{
\mbox{\subfigure[]{\includegraphics[width=8cm,height=6cm]{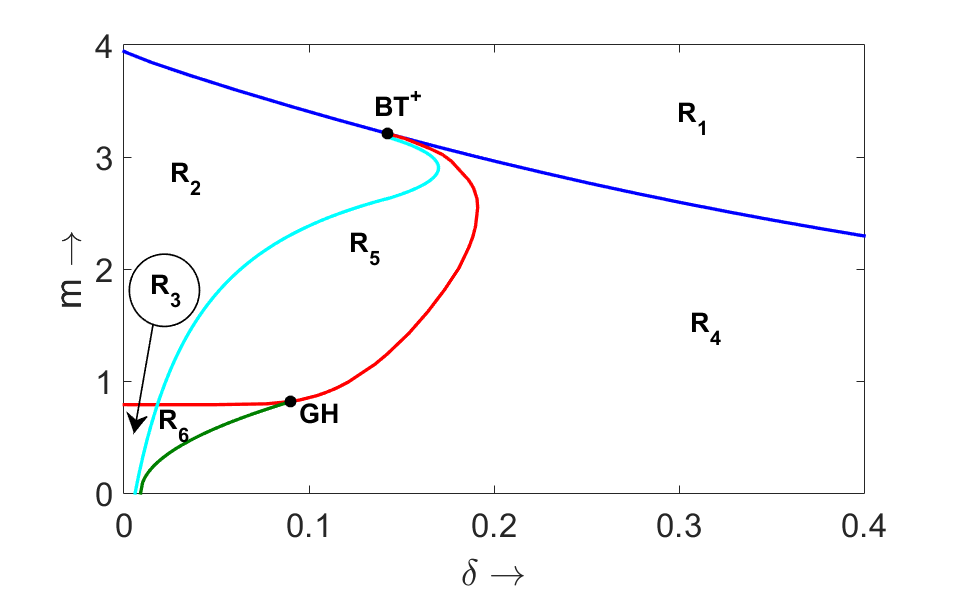}}
\subfigure[]{\includegraphics[width=8cm,height=6cm]{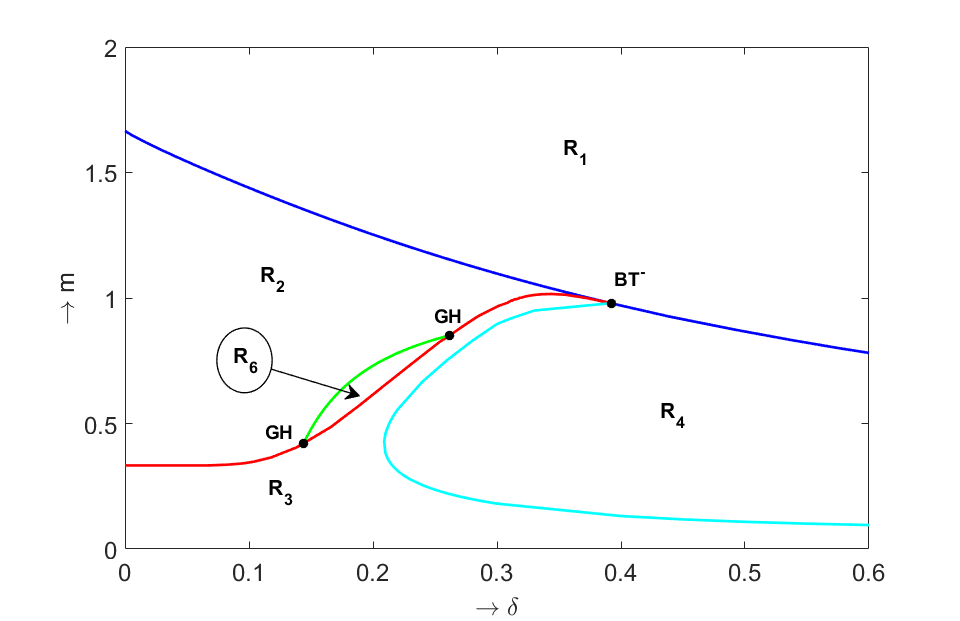}}}}
\caption{S. Bifurcation diagrams in $\delta-m$ parametric plane for the
Ivlev parametrization of Allee effect function. Parameter values are
(a) $\alpha\,=\,7.1$, $\beta\,=\,0.8$; (b) $\alpha\,=\,3$,
$\beta\,=\,0.8$. The dark blue curve is saddle-node bifurcation
curve, the red curve is a Hopf bifurcation curve, the green curve is
the curve of a saddle-node bifurcation of limit cycles. The cyan
curve describes a homoclinic bifurcation.}\label{fig:suppl1}
\end{figure}

\begin{figure}[H]
\centerline{
\mbox{\subfigure[]{\includegraphics[width=8cm,height=6cm]{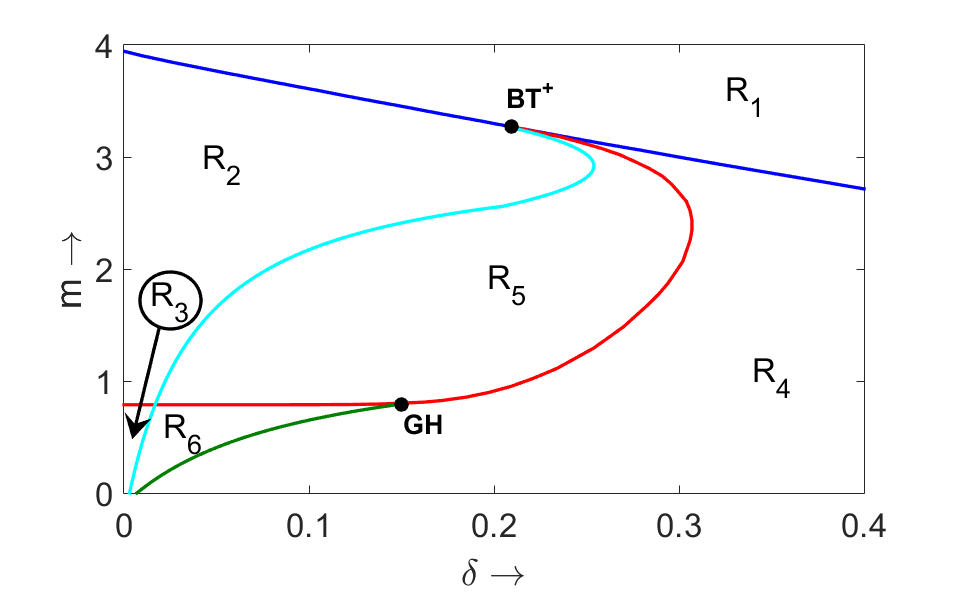}}
\subfigure[]{\includegraphics[width=8cm,height=6cm]{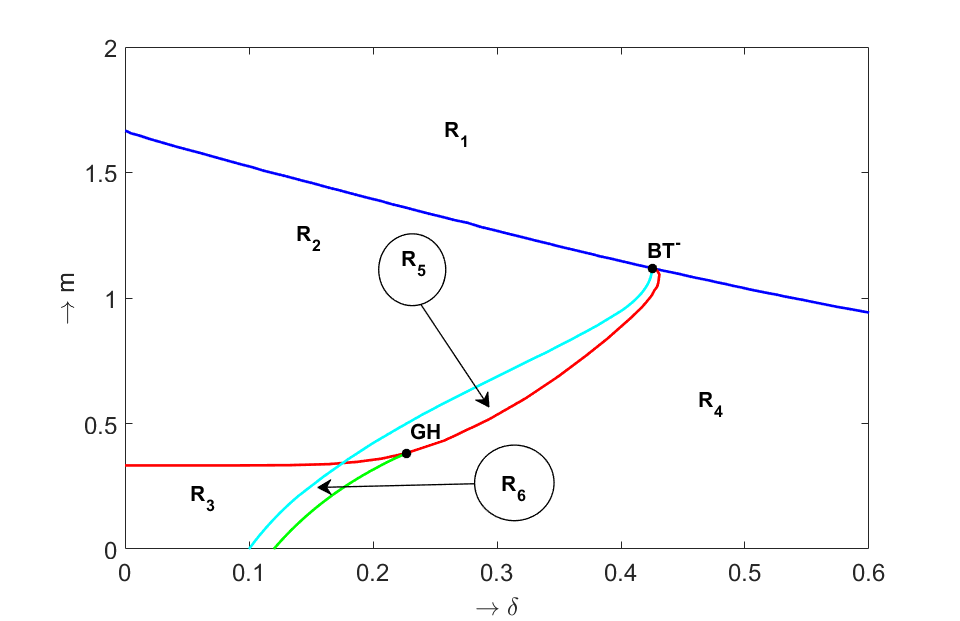}}}}
\caption{S. Bifurcation diagrams in $\delta-m$ parametric plane for
the tan-hyperbolic parametrization of $h(y)$. Parameter values are
(a) $\alpha\,=\,7.1$, $\beta\,=\,0.8$; (b) $\alpha\,=\,3$,
$\beta\,=\,0.8$. The dark blue curve is saddle-node bifurcation
curve, the red curve is a Hopf bifurcation curve, the green curve is
the curve of a saddle-node bifurcation of limit cycles. The cyan
curve describes a homoclinic bifurcation.}\label{fig:suppl2}
\end{figure}

Firstly, we present the bifurcation diagrams for the model with the Ivlev response constructed for $\alpha\,=\,7.1$ and $\alpha\,=\,3$ which is shown in Fig.~1S. The bifurcation diagrams contain saddle-node bifurcation curve (blue), Hopf-bifurcation curve (red) and two global bifurcation curves such as the homoclinic bifurcation curve (cyan) and the saddle-node bifurcation of limit cycles (green). Two local bifurcation curves intersect at the Bogdanov-Takens(BT) bifurcation point and generalised Hopf bifurcation (GH) points are located on the Hopf bifurcation curve. The homoclinic bifurcation curve and the saddle-node bifurcation curve of limit cycles emerge from the BT and GH points, respectively. Note that for $\alpha\,=\,3$, we find two GH points which are connected by a global bifurcation curve. These four curves divide the parametric domains into six and five regions $R_j$ respectively. Regions are described in detail in the main text.

Secondly, we present the bifurcation diagrams constructed for the hyperbolic tangent parametrisation of $h(y)$ which are shown in Figs.~2S. We consider the same values of the parameter $\alpha$. The meaning of the curves and domains are the same as in the previous figure as well as in the main text. One can see that although the mutual position and the shape of bifurcation curves are similar to the Ivlev parameterisations, their shape is slightly different as those in Figs. 1S.

\section*{Supplementary Material 4 (SM4)}

Here we will discuss the dynamics of two prey-predator models consisting the Allee effect in the predator: one with the Allee function in both functional and numerical responses (System \eqref{eq:Allee_both}) and another with the Allee function in numerical response only (System \eqref{eq:Allee_numerical}). First one describes the Allee effect due to collective exploitation of resources and another describes the same due to the lack of mating partner. We consider RM model with linear functional response and choose Monod parametrisation ($h(y)\,=\,\frac{y}{y+\delta}$) for the Allee function. Without loss of generality, we consider dimensionless model and the evolving equations are in the following:

\begin{equation}
\label{eq:Allee_both}
    \begin{split}
    \dot{x} &= x(1-x)-\frac{xy^2}{y+\delta}\\
    \dot{y} &= \frac{\alpha xy^2}{y+\delta}-my
    \end{split}
\end{equation}

\begin{equation}
\label{eq:Allee_numerical}
    \begin{split}
    \dot{x} &= x(1-x)-xy\\
    \dot{y} &= \frac{\alpha xy^2}{y+\delta}-my
    \end{split}
\end{equation}

 To illustrates the dynamics of the above models, we construct bifurcation diagrams for both models in $\delta$-$m$ parametric plane for fixed $\alpha\,=\,8$. Fig.~\ref{fig:suppl3}S present the bifurcation portraits of the model \eqref{eq:Allee_both} consisting three local bifurcation curves: one saddle-node bifurcation curve (blue), two Hopf-bifurcation curves (red) and two non-local bifurcation curves, known as homoclinic bifurcation (cyan). Also the bifurcation diagram consists two BT points, corresponds to Bogdanov-Takens bifurcation of co-dimension 2, where Hopf-bifurcation curve meets tangentially to the saddle-node bifurcation curve global bifurcation curve and a global bifurcation curve is emerging from it. These curves divide the parametric plane into 6 regions ($R_{1}-R_{5}$) with one common region ($R_{2}$). We label this regions as the same way we did for other bifurcation portraits. Secondly we construct the bifurcation diagram of the model \eqref{eq:Allee_numerical} in Fig.~\ref{fig:suppl4}S when Allee function occurs in the numerical response only. This diagram consists same number of local and global bifurcation curves as previous bifurcation diagram (Fig.~\ref{fig:suppl3}S), rather these diagrams looks topologically equivalent to each other. 

\begin{figure}[ht]
    \centering
    \mbox{\subfigure[]{\includegraphics[width=0.45\textwidth]{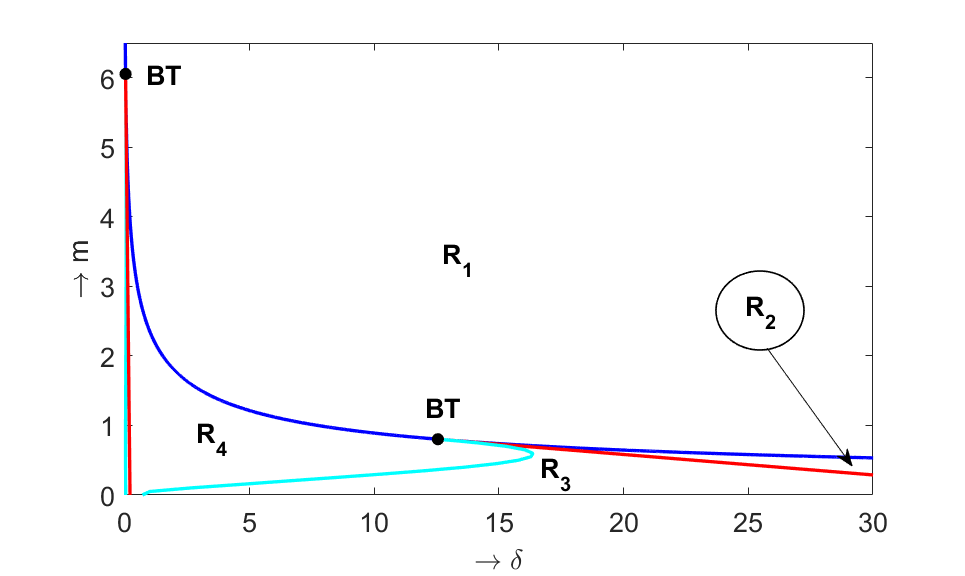}}
    \quad
    \subfigure[]{\includegraphics[width=0.45\textwidth]{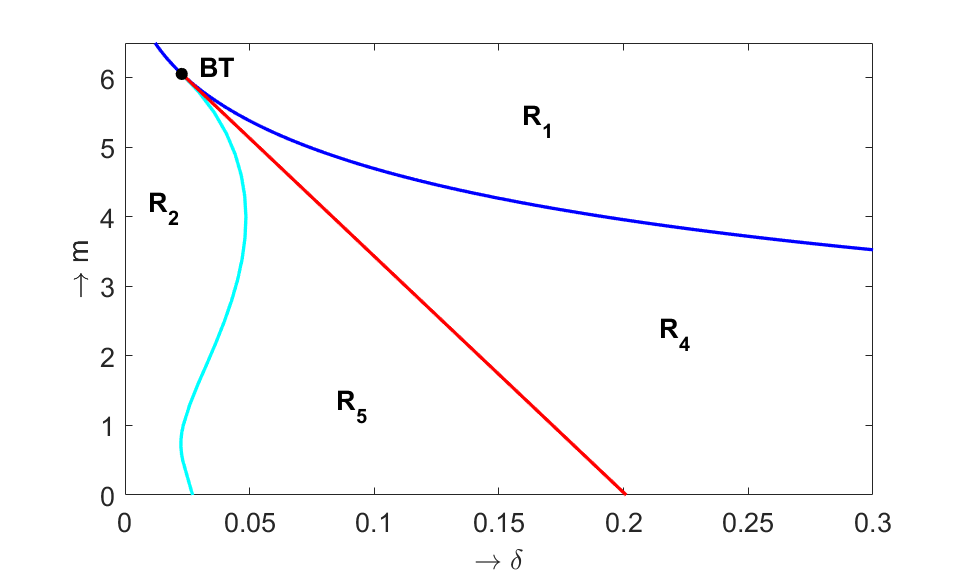}}}
    \caption{S. (a) Two-parametric bifurcation diagram of the system \eqref{eq:Allee_both} in $\delta$-$m$ parametric plane at $\alpha\,=\,8$, (b) zoomed version (a) near origin. Description of all the curves are given in the text.} 
    \label{fig:suppl3}
\end{figure}

\begin{figure}[ht]
    \centering
    \mbox{\subfigure[]{\includegraphics[width=0.45\textwidth]{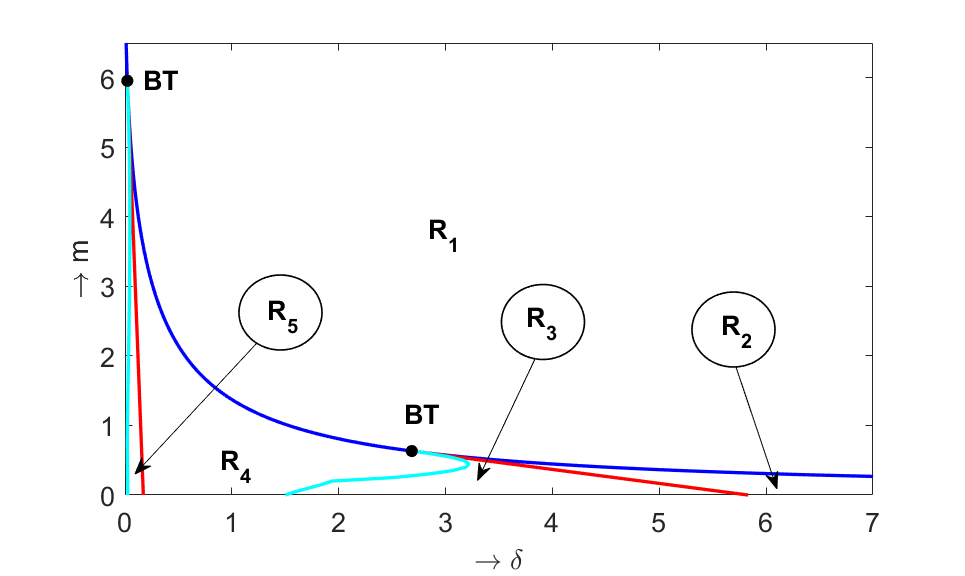}}
    \quad
    \subfigure[]{\includegraphics[width=0.45\textwidth]{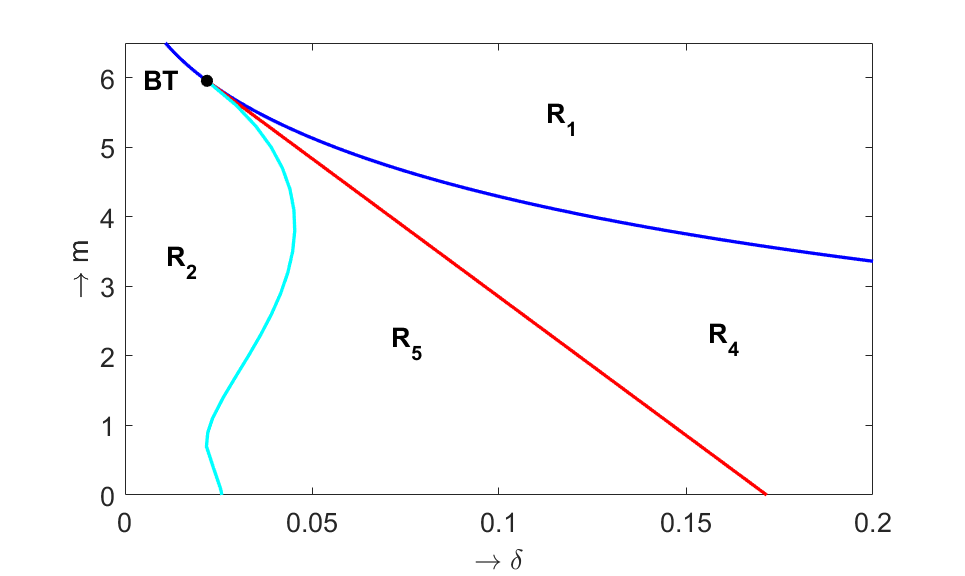}}}
    \caption{S. (a) Two-parametric bifurcation diagram of the system \eqref{eq:Allee_numerical} in $\delta$-$m$ parametric plane at $\alpha\,=\,8$, (b) zoomed version (a) near origin. Description of all the curves are given in the text.} 
    \label{fig:suppl4}
\end{figure}

\begin{figure}[H]
    \centering
    \mbox{\subfigure[]{\includegraphics[width=0.45\textwidth]{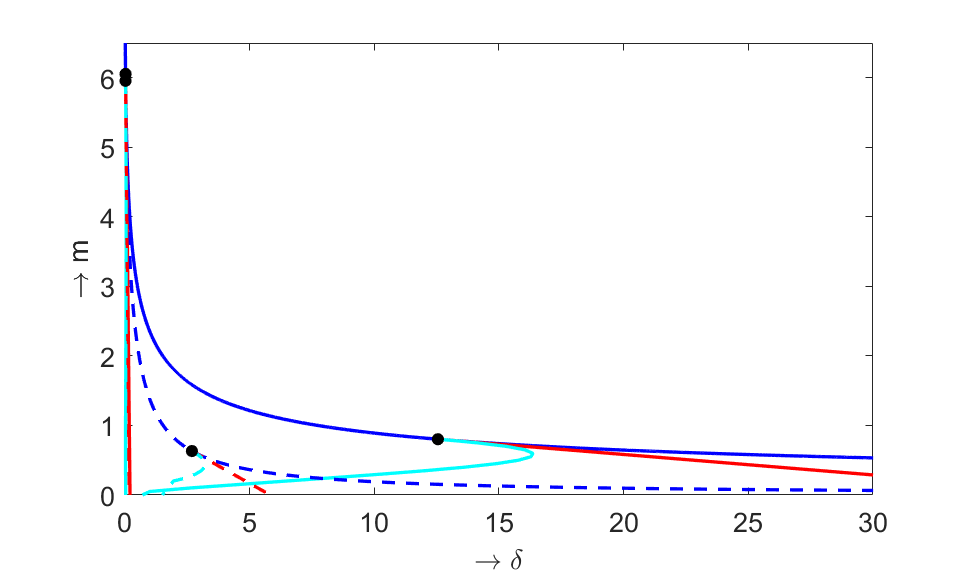}}
    \quad
    \subfigure[]{\includegraphics[width=0.45\textwidth]{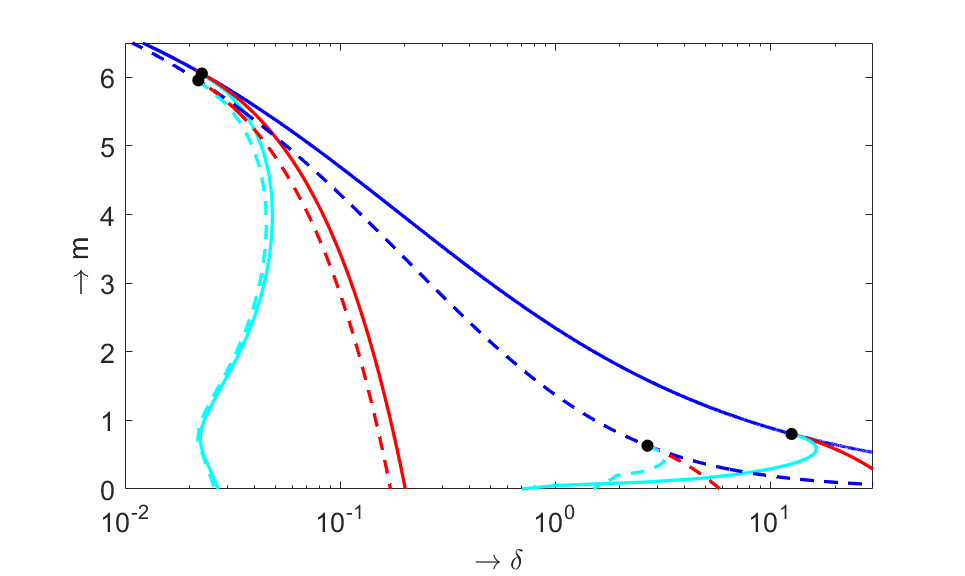}}}
    \caption{S. (a) Combined two-parametric bifurcation diagram of the system \eqref{eq:Allee_both} and \eqref{eq:Allee_numerical} in $\delta$-$m$ parametric plane at fixed $\alpha\,=\,8$, (b) same bifurcation diagram with $x$-axis in logarithmic scale. Solid curves corresponds to system \eqref{eq:Allee_both} and dashed curve corresponds to system \eqref{eq:Allee_numerical}} 
    \label{fig:suppl5} 
\end{figure}

In Fig.~\ref{fig:suppl5}S we have plotted all the local and global bifurcation curves for the two systems compositely to compare their dynamics. We use solid lines for the system \eqref{eq:Allee_both} with Allee function in both functional and numerical responses and dashed line for the other. It is clearly seen that diagrams look like same except for the area of some regions ($R_{i}$) being increased or decreased. The way of appearance of the local and global bifurcation curves are almost same in both cases which leads the diagrams look like similar. But one noticeable fact is that  predator in system \eqref{eq:Allee_both} can persists in larger domain (in $R_{1}$) compared to the system \eqref{eq:Allee_numerical} with Allee function in numerical response only. Thus the mechanism through which Allee effect appears into the predator growth, either due to exploitation of resources or lack of mating partner or something else, has a strong influence on predator persistence and can be considered as separate interesting study.






\end{document}